\documentclass[11pt]{article}
\usepackage[a4paper,margin=2.5cm]{geometry}

\newcommand{\red}[1]{{#1}\xspace}
\newcommand{\markRed}{}

\usepackage{hyperref}
\usepackage{latexsym}
\usepackage{xspace}
\usepackage{amssymb}
\usepackage{amsmath}
\usepackage{stmaryrd}
\usepackage[all]{xy}
\usepackage{calc}
\usepackage{tikz}
\usepackage{flushend}
\usepackage{color}
\usepackage{enumerate}
\usepackage{cite}
\usepackage{authblk}
\usepackage[ruled,vlined,linesnumbered]{algorithm2e}

\newcommand{\email}[1]{\mbox{Email: \url{#1}}}

\def\tuple#1{\langle#1\rangle}
\def\eqref#1{(\ref{#1})}

\newcommand{\E}{\exists}

\def\false{\mathit{false}}
\def\true{\mathit{true}}

\newcommand{\fALC}{$\mathit{f}\!\mathcal{ALC}$\xspace}

\newcommand{\mI}{\mathcal{I}}

\newcommand{\NN}{\mathbb{N}}

\newcommand{\myend}{\mbox{}\hfill{\footnotesize$\blacksquare$}}

\newcommand{\comment}[1]{}

\newcommand{\bbQ}{\mathbb{Q}}

\newcommand{\nekst}{\mathit{next}}
\newcommand{\prev}{\mathit{prev}}

\newcommand{\Null}{\mathit{null}}

\newcommand{\Vector}{\mathit{Vector}}

\newcommand{\CompCB}{\mbox{$\mathsf{ComputeBisimulation}$}\xspace}
\newcommand{\CompCBt}{\mbox{$\mathsf{ComputeBisimulationEfficiently}$}\xspace}

\newcommand{\sCompCB}{\mbox{$\mathsf{s}\textrm{-}\mathsf{ComputeBisimulation}$}\xspace}
\newcommand{\sCompCBt}{\mbox{$\mathsf{s}\textrm{-}\mathsf{ComputeBisimulationEfficiently}$}\xspace}

\newcommand{\Initialize}{\mathsf{Initialize}}
\newcommand{\sInitialize}{\mathsf{s}\textrm{-}\mathsf{Initialize}}

\newcommand{\SV}{\Sigma_V}
\newcommand{\SE}{\Sigma_E}

\newcommand{\bbP}{\mathbb{P}}
\newcommand{\bbS}{\mathbb{S}}
\newcommand{\bbX}{\mathbb{X}}

\newcommand{\itSplit}{\mathit{split}}
\newcommand{\sSplit}{\textit{s-split}}

\newcommand{\pushKey}{\mathit{pushKey}}
\newcommand{\popKey}{\mathit{popKey}}
\newcommand{\maxKey}{\mathit{maxKey}}
\newcommand{\departingBlockEdge}{\mathit{departingBlockEdge}}
\newcommand{\sourceBlockEdge}{\mathit{sourceBlockEdge}}
\newcommand{\BlockEdge}{\mathit{BlockEdge}}
\newcommand{\blockEdge}{\mathit{blockEdge}}
\newcommand{\edgeLabel}{\mathit{label}}
\newcommand{\edgeOrigin}{\mathit{origin}}
\newcommand{\edgeDest}{\mathit{destination}}
\newcommand{\Vertex}{\mathit{Vertex}}
\newcommand{\VertexList}{\mathit{VertexList}}
\newcommand{\Edge}{\mathit{Edge}}
\newcommand{\vertexID}{\mathit{id}}
\newcommand{\vertexBlock}{\mathit{block}}
\newcommand{\Block}{\mathit{Block}}
\newcommand{\BlockList}{\mathit{BlockList}}
\newcommand{\departingBlocks}{\mathit{departingSubblocks}}
\newcommand{\comingEdges}{\mathit{comingEdges}}
\newcommand{\superBlocks}{\mathit{superBlocks}}
\newcommand{\SuperBlock}{\mathit{SuperBlock}}
\newcommand{\vertices}{\mathit{vertices}}
\newcommand{\verticesYp}{\mathit{vertices\!\_of\!\_\!Y'}}
\newcommand{\blocks}{\mathit{blocks}}
\newcommand{\compound}{\mathit{compound}}
\newcommand{\smallerBlock}{\mathit{smallerBlock}}
\newcommand{\Size}{\mathit{size}}
\newcommand{\Partition}{\mathit{Partition}}
\newcommand{\partition}{\mathit{partition}}
\newcommand{\SuperPartition}{\mathit{SuperPartition}}
\newcommand{\superPartition}{\mathit{superPartition}}
\newcommand{\compoundSuperBlocks}{\mathit{compoundSuperBlocks}}
\newcommand{\simpleSuperBlocks}{\mathit{simpleSuperBlocks}}
\newcommand{\degree}{\mathit{degree}}
\newcommand{\add}{\mathit{add}}

\newcommand{\Empty}{\mathit{empty}}
\newcommand{\keys}{\mathit{keys}}
\newcommand{\createBlock}{\mathit{createBlock}}

\newcommand{\createSuperBlock}{\mathit{createSuperBlock}}
\newcommand{\SuperBlockList}{\mathit{SuperBlockList}}
\newcommand{\EdgeList}{\mathit{EdgeList}}
\newcommand{\addBlock}{\mathit{addBlock}}
\newcommand{\addSuperBlock}{\mathit{addSuperBlock}}
\newcommand{\getVertex}{\mathit{getVertex}}
\newcommand{\Not}{\mathit{not}\ }
\newcommand{\blockEdges}{\mathit{blockEdges}}
\newcommand{\edges}{\mathit{edges}}
\newcommand{\clear}{\mathit{clear}}
\newcommand{\removeBlock}{\mathit{removeBlock}}
\newcommand{\block}{\mathit{block}}
\newcommand{\bool}{\mathit{bool}}
\newcommand{\processed}{\mathit{processed}}
\newcommand{\first}{\mathit{first}}
\newcommand{\changed}{\mathit{changed}}

\newcommand{\ComputeBlockEdges}{\mathsf{ComputeBlockEdges}}
\newcommand{\ComputeSubblocks}{\mathsf{ComputeSubblocks}}
\newcommand{\ClearAuxiliaryInfo}{\mathsf{ClearAuxiliaryInfo}}
\newcommand{\Split}{\mathsf{Split}}
\newcommand{\DoSplitting}{\mathsf{DoSplitting}}

\newcommand{\sComputeSubblocks}{\mathsf{s}\textrm{-}\mathsf{ComputeSubblocks}}
\newcommand{\sDoSplitting}{\mathsf{s}\textrm{-}\mathsf{DoSplitting}}
\newcommand{\sSfSplit}{\mathsf{s}\textrm{-}\mathsf{Split}}

\newcommand{\modified}[1]{{#1}\xspace}

\SetKwInput{Input}{Input}
\SetKwInput{Output}{Output}
\SetKwInput{Purpose}{Purpose}
\SetKw{Break}{break}
\SetKw{Continue}{continue}
\SetKw{Return}{return}
\SetKw{New}{new\,}

\newtheorem{theorem}{Theorem}[section]
\newtheorem{lemma}[theorem]{Lemma}
\newtheorem{proposition}[theorem]{Proposition}
\newtheorem{corollary}[theorem]{Corollary}

\newtheorem{Definition}[theorem]{Definition}
\newtheorem{Example}[theorem]{Example}
\newtheorem{Remark}[theorem]{Remark}

\newenvironment{example}{\begin{Example}\begin{em}}{\end{em}\end{Example}}

\newenvironment{proof}{
	
	\smallskip
	
	\noindent
	{\em Proof.}}{
	
	\smallskip
	
}

\begin{document}
\sloppy
	
\title{Computing Crisp Bisimulations for Fuzzy Structures}

\author[1,2]{Linh Anh Nguyen}
\author[3]{Dat Xuan Tran}

\affil[1]{\small Institute of Informatics, University of Warsaw, Banacha 2, 02-097 Warsaw, Poland, \email{nguyen@mimuw.edu.pl}}

\affil[2]{\small
	Faculty of Information Technology, Nguyen Tat Thanh University, Ho Chi Minh City, Vietnam
}

\affil[3]{\small
	\email{dattranx105@gmail.com}
}

\date{}

\maketitle

\begin{abstract}
Fuzzy structures such as fuzzy automata, fuzzy transition systems, weighted social networks and fuzzy interpretations in fuzzy description logics have been widely studied. For such structures, bisimulation is a natural notion for characterizing indiscernibility between states or individuals. There are two kinds of bisimulations for fuzzy structures: crisp bisimulations and fuzzy bisimulations. While the latter fits to the fuzzy paradigm, the former has also attracted attention due to the application of crisp equivalence relations, for example, in minimizing structures. Bisimulations can be formulated for fuzzy labeled graphs and then adapted to other fuzzy structures. In this article, we present an efficient algorithm for computing the partition corresponding to the largest crisp bisimulation of a given finite fuzzy labeled graph. Its complexity is of order $O((m\log{l} + n)\log{n})$, where $n$, $m$ and $l$ are the number of vertices, the number of nonzero edges and the number of different fuzzy degrees of edges of the input graph, respectively. We also study a similar problem for the setting with counting successors, which corresponds to the case with qualified number restrictions in description logics and graded modalities in modal logics. In particular, we provide an efficient algorithm with the complexity $O((m\log{m} + n)\log{n})$ for the considered problem in that setting. 

\medskip

\end{abstract}


\section{Introduction}
\label{section:intro}

\red{Fuzzy structures such as fuzzy automata, fuzzy transition systems, weighted social networks and fuzzy interpretations in fuzzy description logics have been widely studied.}\footnote{\markRed This can be checked by searching Google Scholar.} All of these subjects concern structures that are graph-based and fuzzy/weighted. In such structures, both labels of vertices (states, nodes or individuals) and labels of edges (transitions, connections or roles) can be fuzzified. 

For structures like labeled transition systems or Kripke models, bisimulation~\cite{vBenthem76,vBenthem83,vBenthem84,Park81,HennessyM85} is a natural notion for characterizing indiscernibility between states. For social networks or interpretations in DLs, that notion characterizes indiscernibility between individuals~\cite{KurtoninaR99,Piro2012,BSDL-INS}. When concerning fuzzy structures instead of crisp ones, there are two kinds of bisimulations: crisp bisimulations and fuzzy bisimulations. While the latter fits to the fuzzy paradigm, the former has also attracted attention due to the application of crisp equivalence relations, for example, in minimizing fuzzy structures. 

In~\cite{CaoCK11} Cao {\em et al.} introduced and studied crisp bisimulations for fuzzy transition systems (FTSs). They provided results on composition operations, subsystems, quotients and homomorphisms of FTSs, which are related to bisimulations. 
In~\cite{DBLP:journals/fss/WuD16} Wu and Deng provided logical characterizations of crisp simulations/bisimulations over FTSs via a Hennessy-Milner logic. 
FTSs are already nondeterministic in a certain sense, as they are nondeterministic transition systems when using only the truth values 0 and 1. 
In~\cite{CaoSWC13} Cao {\em et al.} introduced a behavioral distance to measure the behavioral similarity of states in a nondeterministic FTS (NFTS). Such NFTSs are of a higher order than FTSs with respect to nondeterminism, because in an NFTS, for each state $s$ and action $a$, there may be a number of transitions $\tuple{s,a,\mu}$, where each $\mu$ is a fuzzy set of states. The work~\cite{CaoSWC13} studies properties of the introduced behavioral distance, one of which is the connection to crisp bisimulations between NFTSs, which are also introduced in the same article. 
In~\cite{CiricIDB12} {\'C}iri{\'c} {\em et at.} introduced two kinds of fuzzy simulation and four kinds of fuzzy bisimulation for fuzzy automata. Their work studies invariance of languages under fuzzy bisimulations and characterizes fuzzy bisimulations via factor fuzzy automata. 
Ignjatovi{\'c} {\em et at.}~\cite{IgnjatovicCS15} introduced and studied fuzzy simulations and bisimulations between fuzzy social networks in a way similar to~\cite{CiricIDB12}. 
In~\cite{Fan15} Fan introduced fuzzy bisimulations and crisp bisimulations for some fuzzy modal logics under the G\"odel semantics. She provided results on the invariance of formulas under fuzzy/crisp bisimulations and the Hennessy-Milner property of such bisimulations. 
In~\cite{jBSfDL2} Nguyen {\em et al.} defined and studied fuzzy bisimulations and crisp bisimulations for a large class of DLs under the G{\"o}del semantics. Apart from typical topics like invariance (of concepts, TBoxes and ABoxes) and the Hennessy-Milner property, the other topics studied in~\cite{jBSfDL2} are separation of the expressive powers of fuzzy DLs and minimization of fuzzy interpretations. The latter topic was also studied in~\cite{minimization-by-fBS}.  
As shown in~\cite{Fan15} and~\cite{jBSfDL2}, the difference between crisp bisimulation and fuzzy bisimulation with respect to logical characterizations (under the G\"odel semantics) relies on that involutive negation or the Baaz projection operator is used for the former but not for the latter. 

This work \red{studies} the problem of computing crisp bisimulations for fuzzy structures. 

\subsection{\markRed Related Work}

In~\cite{CiricIJD12} {\'C}iri{\'c} {\em et at.} gave an algorithm for computing the greatest fuzzy simulation/bisimulation (of any kind defined in~\cite{CiricIDB12}) between two finite fuzzy automata. 
They did not provide a detailed complexity analysis.  
Following~\cite{CiricIJD12}, Ignjatovi{\'c} {\em et at.}~\cite{IgnjatovicCS15} gave an algorithm with the complexity $O(ln^5)$ for computing the greatest fuzzy bisimulation between two fuzzy social networks, where $n$ is the number of nodes in the networks and $l$ is the number of different fuzzy values appearing during the computation. Later \red{Mici{\'c}} {\em et at.}~\cite{MicicJS18} provided algorithms with the complexity $O(ln^5)$ for computing the greatest right/left invariant fuzzy quasi-order/equivalence of a finite fuzzy automaton, where $n$ is the number of states of the considered automaton and $l$ is the number of different fuzzy values appearing during the computation. These relations are closely related to the fuzzy simulations/bisimulations studied in~\cite{CiricIDB12,CiricIJD12}. Note that, when the G\"odel semantics is used, the mentioned complexity order $O(ln^5)$ can be rewritten to $O((m+n)n^5)$, where $m$ is the number of (non-zero) transitions/connections in the considered fuzzy automata/networks. 
In~\cite{TFS2020} we provided an algorithm with the complexity $O((m+n)n)$ for computing the greatest fuzzy bisimulation between two finite fuzzy interpretations in the fuzzy DL \fALC under the G\"odel semantics, where $n$ is the number of individuals and $m$ is the number of non-zero instances of roles in the given fuzzy interpretations. We also adapted that algorithm for computing fuzzy simulations/bisimulations between fuzzy finite automata and obtained algorithms with the same complexity order. 

In~\cite{DBLP:journals/fss/WuCBD18} Wu {\em et al.} studied algorithmic and logical characterizations of crisp bisimulations for NFTSs~\cite{CaoSWC13}. The logical characterizations are formulated as the Hennessy-Milner property with respect to some logics. They gave an algorithm with the complexity $O(m^2n^4)$ for testing crisp bisimulation (i.e., for checking whether two given states are bisimilar), where $n$ is the number of states and $m$ is the number of transitions in the underlying nondeterministic FTS.

In~\cite{StanimirovicSC2019} Stanimirovi{\'c} {\em et at.} provided algorithms with the complexity $O(n^5)$ for computing the greatest right/left invariant Boolean (crisp) quasi-order matrix of a weighted automaton over an additively idempotent semiring. Such matrices are closely related to crisp simulations. They also provided algorithms with the complexity $O(n^3)$ for computing the greatest right/left invariant Boolean (crisp) equivalence matrix of a weighted automaton over an additively idempotent semiring. Such matrices are closely related to crisp simulations/bisimulations. 

As the background, also recall that Hopcroft~\cite{Hopcroft71} gave an efficient algorithm with the complexity $O(n\log{n})$ for minimizing states in a deterministic (crisp) finite automaton, and Paige and Tarjan~\cite{PaigeT87} gave efficient algorithms with the complexity $O((m+n)\log{n})$ for computing the coarsest partition of a (crisp) finite graph, for both the settings with stability or size-stability. As mentioned in~\cite{PaigeT87}, an algorithm with the same complexity order for the second setting was given earlier by Cardon and Crochemore~\cite{DBLP:journals/tcs/CardonC82}. 

\subsection{\markRed Motivation and Our Contributions}

As far as we know, there were no algorithms directly formulated for computing crisp bisimulations for fuzzy structures like FTSs or fuzzy interpretations in DLs. One can use the algorithm given by Wu {\em et al.}~\cite{DBLP:journals/fss/WuCBD18} for testing crisp bisimulation for a given FTS (as a special case of NFTS), but the complexity $O(m^2n^4)$ is too high (and computing the largest bisimulation is more costly than testing bisimulation). 
One can also try to adapt the algorithms with the complexity $O(n^3)$ given by Stanimirovi{\'c} {\em et at.}~\cite{StanimirovicSC2019} to compute the largest crisp bisimulation of a given finite fuzzy automaton. 

Bisimulations can be formulated for fuzzy labeled graphs and then adapted to other fuzzy structures. In this article, applying the ideas of the Hopcroft algorithm~\cite{Hopcroft71} and the Paige and Tarjan algorithm~\cite{PaigeT87}, we develop an efficient algorithm for computing the partition corresponding to the largest crisp bisimulation of a given finite fuzzy labeled graph. Its complexity is of order $O((m\log{l} + n)\log{n})$, where $n$, $m$ and $l$ are the number of vertices, the number of nonzero edges and the number of different fuzzy degrees of edges of the input graph, respectively. 
If $l$ is bounded by a constant, for example, when $G$ is a crisp graph, then this complexity is of order $O((m+n)\log{n})$. If \mbox{$m \geq n$} then, taking $l = n^2$ for the \modified{worst} case, the complexity is of order $O(m\log^2(n))$. 

We also study a similar problem for the setting with counting successors, which corresponds to the case 
with size-stable partitions for graphs~\cite{PaigeT87}, qualified number restrictions in DLs~\cite{jBSfDL2}, and graded modalities in modal logics~\cite{DBLP:journals/sLogica/Rijke00}. In particular, we provide an efficient algorithm with the complexity $O((m\log{m} + n)\log{n})$ for computing the partition corresponding to the largest crisp bisimulation of a given finite fuzzy labeled graph in the setting with counting successors. When $m \geq n$, this order can be simplified to $O(m\log^2(n))$.

\subsection{\markRed The Structure of This Work}

The rest of this article is structured as follows. In Section~\ref{section: prel}, we provide preliminaries on fuzzy labeled graphs, partitions and crisp bisimulations for such graphs. In Section~\ref{sec: Alg1}, we present the skeleton of our algorithm for the main setting (without counting successors) and prove its correctness. In Section~\ref{sec: Alg1-details}, we give details on how to implement that algorithm and analyze its complexity. Section~\ref{sec: setting-counting-successors} concerns the setting with counting successors. \red{In Section~\ref{sec: experiments}, we report on our implementation of the provided algorithms and the tests we have performed to check the correctness and efficiency of our algorithms and implemented programs. Section~\ref{sec: conc} contains concluding remarks.}

\section{Preliminaries}
\label{section: prel}

A {\em fuzzy labeled graph}, hereafter called a {\em fuzzy graph} for short, is a structure $G = \tuple{V, E, L, \SV, \SE}$, where $V$ is a set of vertices, $\SV$ (respectively, $\SE$) is a set of vertex labels (respectively, edge labels), \mbox{$E: V \times \SE \times V \to [0,1]$} is called the fuzzy set of labeled edges, and $L: V \to (\SV \to [0,1])$ is called the labeling function of vertices. 
It is {\em finite} if all the sets $V$, $\SV$ and $\SE$ are finite. 
Given vertices $x,y \in V$, a vertex label $p \in \SV$ and an edge label $r \in \SE$, $L(x)(p)$ means the degree of that $p$ is a member of the label of~$x$, and $E(x,r,y)$ the degree of that there is an edge from $x$ to $y$ labeled by~$r$. {\markRed The {\em size} of $E$ is defined to be} 
\[ \red{|E| = |\{\tuple{x,r,y} \in V \times \SE \times V: E(x,r,y) > 0\}|.} \]

Recall that a {\em partition} of $V$ is a set of pairwise disjoint non-empty subsets of $V$ whose union is equal to~$V$. Given an equivalence relation $\sim$ on $V$, the partition corresponding to $\sim$ is \mbox{$\{[x]_\sim \mid x \in V\}$}, where $[x]_\sim$ is the equivalence class of $x$ with respect to $\sim$ (i.e., \mbox{$[x]_\sim = \{x' \in V \mid x' \sim x\}$}). 

Let $\bbP$ and $\bbQ$ be partitions of $V$. We say that $\bbP$ is a {\em refinement} of $\bbQ$ if, for every $X \in \bbP$, there exists $Y \in \bbQ$ such that $X \subseteq Y$. In that case we also say that $\bbQ$ is {\em coarser} than $\bbP$. By this definition, every partition is a refinement of itself. 
Given a refinement $\bbP$ of a partition $\bbQ$, a block $Y \in \bbQ$ is {\em compound} with respect to $\bbP$ if there exists $Y' \in \bbP$ such that $Y' \subset Y$. 

Given a fuzzy graph $G = \tuple{V, E, L, \SV, \SE}$, a non-empty binary relation $Z \subseteq V \times V$ is called a {\em crisp auto-bisimulation} of $G$, or a {\em bisimulation} of $G$ for short, if the following conditions hold (for all possible values of the free variables): 
\begin{eqnarray}
Z(x,x') & \to & L(x) = L(x') \label{eq: CB 1} \\[1ex]
Z(x,x') \land E(x,r,y) > 0 & \to & \E y'\,(Z(y,y') \land E(x,r,y) \leq E(x',r,y')) \label{eq: CB 2} \\[1ex]
Z(x,x') \land E(x',r,y') > 0 & \to & \E y\,(Z(y,y') \land E(x',r,y') \leq E(x,r,y)), \label{eq: CB 3}
\end{eqnarray}
where $\to$ and $\land$ denote the usual crisp logical connectives. 

The above definition coincides with the one of~\cite[Section~4.1]{jBSfDL2} for the case when $\mI = \mI'$ and $\Phi = \emptyset$.  

\begin{proposition}\label{prop: MNGOA}
Let $G = \tuple{V, E, L, \SV, \SE}$ be a fuzzy graph. Then, the following assertions hold.
\begin{enumerate}
\item The relation $\{\tuple{x,x} \mid x \in V\}$ is a bisimulation of~$G$.
\item If $Z$ is a bisimulation of $G$, then so is $Z^{-1}$.
\item If $Z_1$ and $Z_2$ are bisimulations of $G$, then so is $Z_1 \circ Z_2$.
\item If $\mathcal{Z}$ is a non-empty set of bisimulations of $G$, then so is~$\bigcup\mathcal{Z}$.
\end{enumerate}
\end{proposition}

The proof of this proposition is straightforward. 
The following corollary immediately follows from this proposition. 

\begin{corollary}\label{cor: JHDKS}
The largest bisimulation of a fuzzy graph exists and is an equivalence relation.
\end{corollary}

Given a fuzzy graph $G = \tuple{V, E, L, \SV, \SE}$, by the {\em partition corresponding to the largest bisimulation of~$G$} we mean the partition of $V$ that corresponds to the equivalence relation being the largest bisimulation of~$G$. 

\section{The Skeleton of the Algorithm}
\label{sec: Alg1}

In this section, we present an algorithm for computing the partition corresponding to the largest bisimulation of a given finite fuzzy graph. It is formulated on an abstract level, without implementation details. The aim is to facilitate understanding the algorithm and prove its correctness. Other aspects of the algorithm are presented in the next section.

In the following, let $G = \tuple{V, E, L, \SV, \SE}$ be a finite fuzzy graph. 
We will use 
$\bbP$, $\bbQ$ and $\bbS$ to denote partitions of $V$, 
$X$ and $Y$ to denote non-empty subsets of $V$, 
and $r$ to denote an edge label from $\SE$. 

For $x \in V$, we denote $E(x,r,Y) = \{E(x,r,y) \mid y \in Y\}$. 
We say that $X$ is {\em stable with respect to $\tuple{Y,r}$} (and $G$) if $\sup E(x,r,Y) = \sup E(x',r,Y)$ for all $x,x' \in X$, where the suprema are taken in the complete lattice [0,1].

We say that a partition $\bbP$ is {\em stable with respect to $\tuple{Y,r}$} (and $G$) if every $X \in \bbP$ is stable with respect to $\tuple{Y,r}$. Next, $\bbP$ is {\em stable} (with respect to $G$) if it is stable with respect to $\tuple{Y,r}$ for all $Y \in \bbP$ and $r \in \SE$. 

\begin{lemma}\label{lemma: JHKSS}
Let $\mathcal{Y}$ be a non-empty family of non-empty subsets of $V$. If a partition $\bbP$ is stable with respect to $\tuple{Y,r}$ for all $Y \in \mathcal{Y}$, then it is also stable with respect to $\tuple{\bigcup\mathcal{Y},r}$.
\end{lemma}

The proof of this lemma is trivial. 

By $\bbP_0$ we denote the partition of $V$ that corresponds to the equivalence relation 
\[ \textrm{$\{\tuple{x,x'} \in V^2 \mid$ $L(x) = L(x')$ and $\sup E(x,r,V) = \sup E(x',r,V)$ for all $r \in \SE\}$.} \] 
 
The following lemma provides another look on the considered problem, leading to a constructive computation.  

\begin{lemma}\label{lemma: HGRLA}
$\bbQ$ is the partition corresponding to the largest bisimulation of $G$ iff it is the coarsest stable refinement of~$\bbP_0$.\footnote{This lemma can be generalized by allowing $G$ to be {\em image-finite} or {\em witnessed}. $G$ is image-finite if, for every $x \in V$ and $r \in \SE$, the set $\{y \mid E(x,r,y) > 0\}$ is finite. $G$ is witnessed if, for every $x \in V$, $r \in \SE$ and $Y \subseteq V$, the set $E(x,r,Y)$ has the biggest element.}  
\end{lemma}

\begin{proof}
It is sufficient to prove the following assertions.
\begin{enumerate}
\item If a partition $\bbQ$ is a stable refinement of~$\bbP_0$, then its corresponding equivalence relation is a bisimulation of~$G$. 
\item If $\bbQ$ is the partition corresponding to the largest bisimulation of $G$, then it is a stable refinement of~$\bbP_0$.  
\end{enumerate}

Consider the first assertion. 
Let $\bbQ$ be a stable refinement of~$\bbP_0$ and $Z$ the equivalence relation corresponding to $\bbQ$. 
We need to show that $Z$ satisfies Conditions \eqref{eq: CB 1}--\eqref{eq: CB 3}. Condition~\eqref{eq: CB 1} holds since $\bbQ$ is a refinement of $\bbP_0$. Consider Condition~\eqref{eq: CB 2} and assume that $Z(x,x')$ and $E(x,r,y) > 0$ hold for some $x,x',y \in V$ and $r \in \SE$. Since $\bbQ$ is stable, $[x]_Z$ is stable with respect to $\tuple{[y]_Z,r}$. Hence, $\sup E(x,r,[y]_Z) = \sup E(x',r,[y]_Z)$, and therefore, $E(x,r,y) \leq \sup E(x',r,[y]_Z)$. Since $G$ is finite, it follows that there exists $y' \in V$ such that $Z(y,y')$ holds and $E(x,r,y) \leq E(x',r,y')$. This completes the proof for Condition~\eqref{eq: CB 2}. Since $Z$ is symmetric, Condition~\eqref{eq: CB 3} is equivalent to Condition~\eqref{eq: CB 2} and also holds. 
	
Consider the second assertion. 
Let $\bbQ$ be the partition corresponding to the largest bisimulation $Z$ of $G$. Due to Conditions~\eqref{eq: CB 1}--\eqref{eq: CB 3}, $\bbQ$ is a refinement of $\bbP_0$. It remains to show that $\bbQ$ is stable. Let $X,Y \in \bbQ$, $x,x'\in X$ and $r \in \SE$. We need to show that $\sup E(x,r,Y) = \sup E(x',r,Y)$. Let $y \in Y$. If $E(x,r,y) = 0$, then clearly $E(x,r,y) \leq \sup E(x',r,Y)$. Suppose that $E(x,r,y) > 0$. By Condition~\eqref{eq: CB 2}, there exists $y'$ such that $Z(y,y')$ holds and $E(x,r,y) \leq E(x',r,y')$. Thus, $y' \in Y$ and $E(x,r,y) \leq \sup E(x',r,Y)$. We have proved that $\sup E(x,r,Y) \leq \sup E(x',r,Y)$. Analogously, it can be proved that $\sup E(x',r,Y) \leq \sup E(x,r,Y)$. 
Hence, $\sup E(x,r,Y) = \sup E(x',r,Y)$, which completes the proof. 
\myend
\end{proof}

\smallskip

In the following, let $\emptyset \subset Y' \subset Y$. 
By $\itSplit(X,\tuple{Y',Y,r})$ we denote the coarsest partition of $X$ such that each of its blocks is stable with respect to both $\tuple{Y',r}$ and $\tuple{Y\setminus Y',r}$. Clearly, that partition exists and is computable.  How to implement the function is left for later. If $\bbX = \itSplit(X,\tuple{Y',Y,r})$ contains more than one block, then we say that $X$ is split into $\bbX$ by $\tuple{Y',Y,r}$ (or by $\tuple{Y',r}$ with respect to $Y$ as the context). 
We also define 
\[ \itSplit(\bbP,\tuple{Y',Y,r}) = \bigcup \{\itSplit(X,\tuple{Y',Y,r}) \mid X \in \bbP\}. \]

Clearly, $\itSplit(\bbP,\tuple{Y',Y,r})$ is the coarsest refinement of $\bbP$ that is stable with respect to both $\tuple{Y',r}$ and $\tuple{Y\setminus Y',r}$. 

\begin{lemma}\label{lemma: DFJWA}
Let a stable partition $\bbS$ be a refinement of $\bbP$, which in turn is a refinement of $\bbQ$. 
Let $Y' \in \bbP$ and $Y \in \bbQ$ be blocks such that $Y' \subset Y$. Then, $\bbS$ is a refinement of $\itSplit(\bbP,\tuple{Y',Y,r})$ for any $r \in \SE$. 
\end{lemma}

\begin{proof}
Since $\bbS$ is a refinement of $\bbP$ and $\bbP$ is a refinement of $\bbQ$, both $Y'$ and $Y \setminus Y'$ are unions of a number of blocks of $\bbS$. Since $\bbS$ is stable, it is stable with respect to $\tuple{B,r}$ for all blocks $B \in \bbS$. By Lemma~\ref{lemma: JHKSS}, it follows that $\bbS$ is stable with respect to both $\tuple{Y',r}$ and $\tuple{Y\setminus Y',r}$. Since $\itSplit(\bbP,\tuple{Y',Y,r})$ is the coarsest refinement of $\bbP$ that is stable with respect to both $\tuple{Y',r}$ and $\tuple{Y\setminus Y',r}$, it follows that $\bbS$ is a refinement of $\itSplit(\bbP,\tuple{Y',Y,r})$. 
\myend
\end{proof}

We provide Algorithm~\ref{algCompCB} (on page \pageref{algCompCB}) for computing the partition corresponding to the largest bisimulation of~$G$. It starts by initializing $\bbP$ to $\bbP_0$. If $\bbP$ is a singleton, then $\bbP$ is stable and the algorithm returns it as the result. Otherwise, the algorithm repeatedly refines $\bbP$ to make it stable as follows. The algorithm maintains a partition $\bbQ_r$ of $V$, for each $r \in \SE$, such that $\bbP$ is a refinement of $\bbQ_r$ and $\bbP$ is stable with respect to $\tuple{Y,r}$ for all $Y \in \bbQ_r$. 
If at some stage $\bbQ_r = \bbP$ for all $r \in \SE$, then $\bbP$ is stable and the algorithm terminates with that $\bbP$. The variables $\bbQ_r$ are initialized to $\{V\}$ for all $r \in \SE$. In each iteration of the main loop, the algorithm chooses $\bbQ_r \neq \bbP$, $Y \in \bbQ_r$ and $Y' \in \bbP$ such that $Y' \subset Y$ and $|Y'| \leq |Y|/2$, then it replaces $\bbP$ with \mbox{$\itSplit(\bbP, \tuple{Y',Y,r})$} and replaces $Y$ in $\bbQ_r$ with $Y'$ and $Y \setminus Y'$. In this way, the chosen $\bbQ_r$ is refined (and $\bbP$ may also be  refined), so the loop will terminate after a number of iterations. 
The condition $|Y'| \leq |Y|/2$ reflects the idea ``process the smaller half (or a smaller component)'' from Hopcroft's algorithm~\cite{Hopcroft71} and Paige and Tarjan's algorithm~\cite{PaigeT87}. Without using it the algorithm still terminates with a correct result, but the condition is essential for reducing the complexity order of the algorithm. 

\begin{algorithm}[t]
\caption{\CompCB\label{algCompCB}}
\Input{a finite fuzzy graph $G = \tuple{V, E, L, \SV, \SE}$.}
\Output{the partition corresponding to the largest bisimulation of~$G$.}

\smallskip

let $\bbP = \bbP_0$ and $\bbQ_r = \{V\}$ for all $r \in \SE$\;
\lIf{$\bbP = \{V\}$}{\Return $\bbP$} 
\While{there exists $r \in \SE$ such that $\bbQ_r \neq \bbP$}{ 
	choose such an $r$ and a compound block $Y \in \bbQ_r$ with respect to $\bbP$\;
	choose a block $Y' \in \bbP$ such that $Y' \subset Y$ and $|Y'| \leq |Y|/2$\;
	$\bbP := \itSplit(\bbP, \tuple{Y',Y,r})$\label{step: UIIJW 6}\;
	refine $\bbQ_r$ by replacing $Y$ with $Y'$ and $Y \setminus Y'$\label{step: UIIJW 7}\;
}
\red{\Return $\bbP$}\;
\end{algorithm}


\newcommand{\distA}{1.5cm}
\newcommand{\distB}{2.0cm}
\newcommand{\graphA}{
	\begin{center}
		\begin{tikzpicture}[->,>=stealth]
		\node (c) {$c$};
		\node (d) [node distance=\distA, right of=c] {$d$};
		\node (e) [node distance=\distA, right of=d] {$e$};
		\node (f) [node distance=\distA, right of=e] {$f$};
		\node (g) [node distance=\distA, right of=f] {$g$};
		\node (a) [node distance=\distB, above of=d] {$a$};
		\node (b) [node distance=\distB, above of=f] {$b$};
		\node (x) [node distance=0.3cm, above of=a] {};
		\draw (a) to node [right]{\footnotesize{1}} (c);
		\draw (a) to node [right,yshift=-3mm]{\footnotesize{0.7}} (d);
		\draw (a) to node [right]{\footnotesize{1}} (e);
		\draw (b) to node [left,yshift=2mm,xshift=1.5mm]{\footnotesize{1}} (e);
		\draw (b) to node [left]{\footnotesize{1}} (f);
		\draw (b) to node [left,yshift=-2mm,xshift=1mm]{\footnotesize{0.6}} (g);
		\path (f) edge [left] node [left,yshift=1.8mm,xshift=2mm]{\footnotesize{1}} (a);
		\path (c) edge [bend left=40] node [left]{\footnotesize{1}} (a);
		\path (g) edge [bend right=40] node [right]{\footnotesize{1}} (b);
		\end{tikzpicture}
	\end{center}	
}


\begin{example}\label{example: ill-1}
	Consider the fuzzy graph $G$ illustrated below and specified as $G = \tuple{V, E, L, \SV, \SE}$, 
	where $\SV = \emptyset$, $\SE = \{r\}$, $V = \{a,\ldots,f\}$, $L$ is the empty labeling function (which labels each vertex with the empty fuzzy set), and $E$ is specified by the edges and their fuzzy degrees displayed in the picture. 
	
	\medskip
	
	\graphA
	
	\medskip
	
	Consider the run of Algorithm~\ref{algCompCB} on this fuzzy graph. 
	\begin{itemize}
		\item At the beginning, we have that 
		$\bbP = \bbP_0 = \{\{a,b,c,f,g\}$, $\{d,e\}\}$ and $\bbQ_r = \{V\}$. 
		
		\item During the first iteration of the main loop, $Y = V$ and $Y' = \{d,e\}$. 
		As the effects of this iteration, $\bbP$ is refined to $\{\{a,b\}$, $\{c,f,g\}$, $\{d,e\}\}$, and $\bbQ_r$ is changed to $\bbP_0$. 
		
		\item During the second iteration of the main loop, we have that $Y = \{a,b,c,f,g\}$ and $Y' = \{a,b\}$. 
		$\bbP$ is not further refined, but $\bbQ_r$ is changed to $\bbP$. 
		\item The algorithm terminates with the result $\{\{a,b\}$, $\{c,f,g\}$, $\{d,e\}\}$. 
		\myend
	\end{itemize}
\end{example}

Some properties of Algorithm~\ref{algCompCB} are stated below.

\begin{lemma}\label{lemma: invariants}
Let $\bbS$ be the coarsest stable refinement of~$\bbP_0$. 
Consider Algorithm~\ref{algCompCB} and suppose that $\bbP_0 \neq \{V\}$. 	
The following assertions are invariants of the main loop of the algorithm:
\begin{enumerate}
\item $\bbP$ is a refinement of $\bbP_0$ and $\bbQ_r$ for all $r \in \SE$;
\item $\bbS$ is a refinement of $\bbP$;
\item $\bbP$ is stable with respect to $\tuple{B,r}$ for all $r \in \SE$ and $B \in \bbQ_r$.
\end{enumerate}
\end{lemma}

By Corollary~\ref{cor: JHDKS} and Lemma~\ref{lemma: HGRLA}, the coarsest stable refinement of~$\bbP_0$ exists. 
The first and third invariants clearly hold. The second invariant follows from Lemma~\ref{lemma: DFJWA} and the first invariant.

\begin{theorem}\label{theorem: UDFND}
Algorithm~\ref{algCompCB} always terminates with a correct result.
\end{theorem}

\begin{proof}
It is easy to see that, if $\bbP_0 = \{V\}$, then $\{V\}$ is the partition corresponding to the largest bisimulation of~$G$ and the assertion of the theorem holds. Assume that $\bbP_0 \neq \{V\}$. 

By the first assertion of Lemma~\ref{lemma: invariants}, it is an invariant of the algorithm that $\bbP$ is a refinement of $\bbQ_r$ for all $r \in \SE$. The loop of the algorithm must terminate because the partitions $\bbQ_r$ (for $r \in \SE$) cannot be refined forever. 
	
Let $\bbS$ be the coarsest stable refinement of~$\bbP_0$. 
By the second assertion of Lemma~\ref{lemma: invariants}, $\bbS$ is a refinement of the final $\bbP$. 

By the first assertion of Lemma~\ref{lemma: invariants}, the final $\bbP$ is a refinement of $\bbP_0$. 
At the end of the algorithm, $\bbP = \bbQ_r$ for all $r \in \SE$. Hence, by the third assertion of Lemma~\ref{lemma: invariants}, the final $\bbP$ is stable. 
Thus, the final $\bbP$ is a stable refinement of $\bbP_0$. Therefore, it is a refinement of $\bbS$. Together with the assertion in the above paragraph, this implies that the final $\bbP$ is equal to $\bbS$. 
By Lemma~\ref{lemma: HGRLA}, it follows that $\bbP$ is the partition corresponding to the largest bisimulation of~$G$. This completes the proof.
\myend
\end{proof}

\section{Implementation Details and Complexity Analysis}
\label{sec: Alg1-details}

In this section, we show how to implement Algorithm~\ref{algCompCB} so that its complexity is of order $O((m\log{l} + n)\log{n})$, where $n$, $m$ and $l$ are the number of vertices, the number of nonzero edges and the number of different fuzzy degrees of edges of the input graph $G$, respectively. 
Apart from the mentioned idea ``process a smaller component'', another key for getting that complexity order is to efficiently process the operation $\itSplit(\bbP, \tuple{Y',Y,r})$ at the statement~\ref{step: UIIJW 6} of Algorithm~\ref{algCompCB}. 
Like the Hopcroft algorithm~\cite{Hopcroft71} and the Paige and Tarjan algorithm~\cite{PaigeT87}, for that operation we also start from the vertices from $Y'$ and look backward through the edges coming to them, without scanning $Y$. The processing is, however, quite sophisticated. To enable full understanding of the implementation and its complexity analysis, we use the object-oriented approach and decide to describe the data structures in detail.

\subsection{Data Structures}
\label{sec: DS}

Algorithm~\ref{algCompCB} was formulated on an abstract level. Using the object-oriented approach, we describe how to get an efficient implementation of this algorithm by using a number of classes. In the description given below, we refer to the input graph $G = \tuple{V, E, L, \SV, \SE}$ and the variables $\bbP$ and $\bbQ_r$ ($r \in \SE$) used in the algorithm, which represent partitions of $V$. The classes are listed below:
\begin{itemize}
\item $\Vertex$: the type for the vertices of $G$; 
\item $\Edge$: the type for the edges of $G$; 
\item $\Block$: the type for the blocks of $\bbP$; 
\item $\SuperBlock$: the type for the blocks of $\bbQ_r$ ($r \in \SE$);
\item \red{$\VertexList$ and $\SuperBlockList$: the types for doubly linked lists of elements of the type $\Vertex$ or $\SuperBlock$, respectively;} 
\item \red{$\EdgeList$ and $\BlockList$: the types for lists of elements of the type $\Edge$ or $\Block$, respectively;} 
\item $\Partition$: the type for $\bbP$, defined as $\BlockList$;
\item $\SuperPartition$: the type for $\bbQ_r$ ($r \in \SE$);
\item $\BlockEdge$: the type for objects specifying information about edges connecting a vertex to a block of $\bbQ_r$.
\end{itemize}

We call objects of the type $\SuperBlock$, $\SuperPartition$ or $\BlockEdge$ super-blocks, super-partitions and block-edges, respectively. 
We give below details for nontrivial classes in the above list. 
As in the Java language, attributes of objects are primitive values or references. 

\bigskip

\noindent
{\bf Vertex}. This class has the following \red{instance} attributes. 
\begin{itemize}
\item $\vertexID$ is the ID of the vertex (a natural number or a string).  
\item \mbox{$\vertexBlock: \Block$} is the block of $\bbP$ that contains the vertex.  
\item \red{\mbox{$\nekst: \Vertex$ and $\prev: \Vertex$} are the next vertex and the previous vertex in the doubly linked list that contains the current vertex.}
\item \mbox{$\comingEdges: \SE \to \EdgeList$} is \red{a map that associates each $r \in \SE$ with a list of all edges labeled by $r$ that come to the vertex}. 
\item \mbox{$\processed: \bool$} is a flag for internal processing. 
\end{itemize}
The constructor $\Vertex(id')$ sets $id$ to $id'$, $\vertexBlock$, \red{$\nekst$ and $\prev$} to $\Null$, $\comingEdges$ to a newly created empty \red{map}, and $\processed$ to $\false$. The class also has a static method $\getVertex(id)$ that returns the vertex with the given ID. It uses a \red{static} attribute to store the collection of the vertices that have been created. 

\bigskip

\noindent
{\bf Edge}. This class has the following \red{instance} attributes. 
\begin{itemize}
\item $\edgeLabel: \SE$
\item $\edgeOrigin: \Vertex$
\item $\edgeDest: \Vertex$
\item \mbox{$\degree: [0,1]$} is the value of $E(x,r,y)$, where $x$, $r$ and $y$ are the $\edgeOrigin$, $\edgeLabel$ and $\edgeDest$ of the edge, respectively. 
\item \mbox{$\blockEdge: \BlockEdge$} specifies information about the set of edges labeled by $r$ from $x$ to the vertices of $Y$, where $r$ and $x$ are the $\edgeLabel$ and $\edgeOrigin$ of the edge, respectively, and $Y$ is the block of $\bbQ_r$ that contains the $\edgeDest$ of the edge (via a block of $\bbP$). 
\end{itemize}
The constructor $\Edge(r,x,y,d,bE)$ sets the above listed attributes to the parameters, respectively, and then adds the current edge to the list $\edgeDest.\comingEdges\red{[r]}$. 

\bigskip

\noindent
{\bf BlockEdge}. 
As mentioned above, an object of this class gives information about the set of edges with a label $r$ from a vertex $x$ to a block $Y \in \bbQ_r$. It is defined as an extended map of type \mbox{$[0,1] \to \NN$}, whose keys are the values of $E(x,r,y)$ for $y \in Y$. The value of a key $d \in [0,1]$ in the map is the number of vertices $y \in Y$ such that $E(x,r,y) = d$. Apart from the map, the class has two \red{instance} attributes of type $\BlockEdge$, with names and descriptions given below. 
\begin{itemize}
	\item $\departingBlockEdge$: When the block $Y$ of $\bbQ_r$ is going to be replaced by $Y\setminus Y'$ and $Y'$, the current block-edge (i.e., the object {\em this}) changes to a block-edge with the destination $Y\setminus Y'$, a new block-edge with the destination $Y'$ is created, and the attribute $\departingBlockEdge$ of the current block-edge is set to that new block-edge.   
	\item $\sourceBlockEdge$: This attribute is a converse of $\departingBlockEdge$. That is, the current block-edge is equal to the attribute $\departingBlockEdge$ of the object $\sourceBlockEdge$ if they are set. 
\end{itemize}
Apart from the get/set methods for the above attributes, the class $\BlockEdge$ also has the following methods. 
\begin{itemize}
	\item $\pushKey(d)$: This method increases the value of the key $d$ in the map by~1. If the key is absent, it is added to the map and its value is set to~1. 
	\item $\popKey(d)$: This method decreases the value of the key $d$ in the map by~1, under the assumption that the key is present. If the value becomes 0, then the key is deleted from the map. 
	\item $\maxKey()$: This method returns the biggest key of the map if the map is not empty, and 0 otherwise. 
\end{itemize}
The default constructor $\BlockEdge()$ creates an empty map and sets the additional attributes to $\Null$. 
The constructor $\BlockEdge(bE)$ differs from the default in that it also sets $\sourceBlockEdge$ to $bE$. 

\bigskip

\noindent
{\bf Block}. Objects of this class are the blocks of $\bbP$ (the current partition of $V$ used in the algorithm). The class has the following \red{instance} attributes.  
\begin{itemize}
\item \mbox{$\vertices: \VertexList$} is \red{a} list of vertices of the block. 
\item \mbox{$\partition: \Partition$} is a reference to $\bbP$. 
\item \mbox{$\superBlocks: \SE \to \SuperBlock$} is a map that associates each $r \in \SE$ with the block of $\bbQ_r$ that contains the current block. 
\item \mbox{$\departingBlocks_1: [0,1] \to \VertexList$} and  
\mbox{$\departingBlocks_2: [0,1] \to \VertexList$} are maps whose keys' values specify the splitting to be done for the current block. They are described in more detail below.
\end{itemize}

Let $X$ denote the $\vertices$ of the current block. Consider the statement~\ref{step: UIIJW 6} of Algorithm~\ref{algCompCB} and let $\bbX = \itSplit(X,\tuple{Y',Y,r})$. For each $X_i \in \bbX$, let $d_{i,1} = \sup E(x,r,Y \setminus Y')$ and $d_{i,2} = \sup E(x,r,Y')$ for any $x \in X_i$ (the choice of $x$ does not matter, since $X_i$ is stable with respect to both $\tuple{Y\setminus Y', r}$ and $\tuple{Y', r}$). As an invariant of the algorithm, $X$ is stable with respect to $\tuple{Y,r}$, and hence $\max(d_{i,1}, d_{i,2})$ does not depend on the choice of $X_i$ from $\bbX$. The maps $\departingBlocks_1$ and $\departingBlocks_2$ are set so that, if $d_{i,1} \geq d_{i,2}$, then $\departingBlocks_2[d_{i,2}]$ is a list representing the set $X_i$, else $\departingBlocks_1[d_{i,1}]$ is a list representing the set $X_i$. The computation of these maps will be specified later. 

\red{The class $\Block$ has the following constructor.}

\medskip

\begin{algorithm}[H]
\SetAlgoVlined
\SetKwProg{constructor}{Constructor}{:}{}

\constructor{$\Block(\vertices'$, $\partition'$, $\superBlocks')$\label{contrBlock}}{
$\vertices := \vertices'$, $\partition := \partition'$\;
$\partition.\add(this)$\;
set the attributes $\superBlocks$, $\departingBlocks_1$ and $\departingBlocks_2$ to newly created empty maps of appropriate types\;
\red{\ForEach{$x \in \vertices$}{
	$x.\block := this$\;
}}
\If{$\superBlocks' \neq \Null$}{
	\ForEach{$r \in \SE$}{
		$\superBlocks[r] := \superBlocks'[r]$\;
		$\superBlocks[r].\addBlock(this)$\;
	}
}
} 
\end{algorithm}

\medskip

Let the class $\Block$ have the following static method, whose parameter is a key's value of one of the maps $\departingBlocks_1$ and $\departingBlocks_2$ of a block of~$\bbP$.

\medskip

\begin{algorithm}[H]
\SetAlgoVlined
\SetKwProg{StaticMethod}{Static method}{:}{}
	
\StaticMethod{$\createBlock(vl: \VertexList)$}{\markRed  
	let $x$ be the first element of $vl$\;
	$bx := x.\vertexBlock$\;
	$\New \Block(vl, bx.\partition, bx.\superBlocks)$\;
}
\end{algorithm}

\bigskip

\noindent
{\bf SuperBlock}. Each object of this class represents a block of $\bbQ_r$ for some $r \in \SE$ ($\bbQ_r$ is used in the algorithm as a partition of $V$). It consists of a number of blocks of $\bbP$. The class has the following members. 
\begin{itemize}
\item \mbox{$\blocks: \BlockList$} is a list of the blocks of $\bbP$ that compose the current super-block. 
\item \mbox{$\superPartition: \SuperPartition$} is a reference to $\bbQ_r$. 
\item \red{\mbox{$\nekst: \SuperBlock$ and $\prev: \SuperBlock$} are the next super-block and the previous super-block in the doubly linked list that contains the current super-block;}
\item $SuperBlock(superPartition')$ is the constructor that initializes $\blocks$ to a newly created empty list, sets $\superPartition$ to $superPartition'$, \red{$\nekst$ and $\prev$ to $\Null$,} and adds the current super-block to $\superPartition$ by calling $\superPartition.\addSuperBlock(this)$. 
\item $\Size()$ is the method that returns $\blocks.\Size()$ (the number of blocks of $\bbP$ that compose the current super-block). 
\item $\compound()$ is the method that returns the truth of $\Size() > 1$. That is, this Boolean method returns $true$ iff the current super-block is compound with respect to~$\bbP$. 
\item $\smallerBlock()$ is a method that can be called only when the current super-block is compound. It compares the first two blocks of the current super-block and returns the smaller one (or any one when their sizes are equal).
\item $\addBlock(b)$ is the method that adds the block $b$ to the list $\blocks$. If the addition causes that $\Size() = 2$, then the method also moves the current super-block from $\superPartition.\simpleSuperBlocks$ to $\superPartition.\compoundSuperBlocks$. 
\item $\removeBlock(b)$ is the method that removes the block $b$ from the list $\blocks$. If the removal causes that $\Size() = 1$, then the method also moves the current super-block from $\superPartition.\compoundSuperBlocks$ to $\superPartition.\simpleSuperBlocks$. \red{This method will be called only for $b$ being the result returned by the method $\smallerBlock()$ and therefore can be executed in constant time.}
\item $\createSuperBlock(sp,b,r)$ is a static method that creates a new super-block $sb$ for the super-partition $sp$, adds the block $b$ to $sb$ and makes $sb$ the super-block of $b$ in $\bbQ_r$. 
It consists of the statements $sb := \New \SuperBlock(sp)$, $sb.\addBlock(b)$, and \mbox{$b.\superBlocks[r] := sb$}.
\end{itemize}

\medskip

\noindent
{\bf SuperPartition}. An object of this class represents $\bbQ_r$ for some $r \in \SE$ ($\bbQ_r$ is used in the algorithm as a partition of $V$). It consists of a number of super-blocks (i.e., objects of type $\SuperBlock$). The class has the following \red{instance} attributes. 
\begin{itemize}
\item \mbox{$\compoundSuperBlocks: \SuperBlockList$} is a list consisting of all the compound super-blocks (each of which consists of more than one block).  
\item \mbox{$\simpleSuperBlocks: \SuperBlockList$} is a list consisting of all the simple super-blocks (each of which contains at most one block). 
\end{itemize}
The constructor $\SuperPartition()$ initializes the above mentioned attributes to newly created empty lists. The class has the method $\addSuperBlock(sb)$, which adds the super-block $sb$ to the list $\compoundSuperBlocks$ or $\simpleSuperBlocks$ depending on whether $sb$ is compound or not.

\subsection{Initialization}
\label{sec: init}

Our revision of Algorithm~\ref{algCompCB} uses Procedure $\Initialize$ (on page~\pageref{procInitialize}), which sets up the global variables $\bbP$ and $\bbQ$, where $\bbQ$ is a map of type \mbox{$\SE \to \SuperPartition$} and $\bbQ[r]$ means $\bbQ_r$. 

\begin{procedure}
\caption{Initialize(G)\label{procInitialize}} 

construct \red{a vector \mbox{$\vertices: \Vector(\Vertex)$} that contains all} vertices of $G$ by using the constructor $\Vertex(id)$\label{step: HGJWA 1}\; 
construct \red{a vector \mbox{$\edges: \Vector(\Edge)$} that contains all} edges of $G$ by calling \red{the static method $\Vertex.\getVertex(id)$} and \red{the constructor} $\Edge(r,x,y,d,\Null)$ appropriately (this also sets up the lists of coming edges for the vertices)\label{step: HGJWA 2}\;

$\bbP := \New  \Partition()$\label{step: HGJWA 3}\;
create an empty map $\bbQ: \SE \to \SuperPartition$\label{step: HGJWA 4}\;
\red{create an empty map $msb : \SE \to \SuperBlock$\label{step: HGJWA 5}\;}
\ForEach{$r \in \SE$\label{step: HGJWA 6}}{
	$\bbQ[r] := \New \SuperPartition()$\label{step: HGJWA 7}\;
	\red{$msb[r] := \SuperBlock(\bbQ[r])$\label{step: HGJWA 8}}\;
}
create an empty map $\blockEdges: \Vertex \times \SE \to \BlockEdge$\label{step: HGJWA 11}\; 
\ForEach{$x$ in $\vertices$ and $r \in \SE$\label{step: HGJWA 12}}{
	$\blockEdges[x,r] := \New \BlockEdge()$\label{step: HGJWA 13}\; 
}
\ForEach{$e$ in $\edges$\label{step: HGJWA 14}}{
	$e.\blockEdge := \blockEdges[e.\edgeOrigin, e.\edgeLabel]$\label{step: HGJWA 15}\;
	$e.\blockEdge.\pushKey(e.\degree)$\label{step: HGJWA 16}\;
}
 
create a set $\bbX$ of objects of type $\VertexList$ such that, for every $x$ and $x'$ in $\vertices$, there exists $X \in \bbX$ such that both $x$ and $x'$ belong to $X$ iff $L(x.\vertexID) = L(x'.\vertexID)$ and $\blockEdges[x,r].\maxKey() = \blockEdges[x',r].\maxKey()$ for all $r \in \SE$\label{step: HGJWA 17}\; 

\lForEach{$X \in \bbX$\label{step: HGJWA 18}}{\red{$\Block(X, \bbP, msb)$}}
\end{procedure}

Let's analyze the complexity of Procedure $\Initialize$. 
Recall that the sizes of $\SE$ and $\SV$ are assumed to be bounded by a constant. 
The time needed for running the steps is as follows: 
\ref{step: HGJWA 1}:~$O(n\log{n})$;  
\ref{step: HGJWA 2}:~$O(m\log{n})$;  
\mbox{\ref{step: HGJWA 3}--\ref{step: HGJWA 11}}:~$O(1)$;   
\mbox{\ref{step: HGJWA 12}--\ref{step: HGJWA 13}}:~$O(n \log{n})$;   
\mbox{\ref{step: HGJWA 14}--\ref{step: HGJWA 16}}:~$O(m \log{n})$;   
\ref{step: HGJWA 17}:~$O(n \log{n})$; 
and \ref{step: HGJWA 18}:~$O(n)$. 
Thus, the time complexity of the procedure $\Initialize(G)$ is of order \mbox{$O((m+n)\log{n})$}. 

\subsection{The Revised Algorithm}

We revise Algorithm~\ref{algCompCB} to obtain Algorithm~\ref{algCompCBt} (on page~\pageref{algCompCBt}), which uses the classes specified in Section~\ref{sec: DS} and the procedure $\Initialize(G)$ given in Section~\ref{sec: init}. The new algorithm uses global variables $\bbP$ and $\bbQ$ for the subroutines, where $\bbQ$ is a map of type \mbox{$\SE \to \SuperPartition$}. $\bbP$ and $\bbQ[r]$ (for $r \in \SE$) correspond to the variables $\bbP$ and $\bbQ_r$ used in Algorithm~\ref{algCompCB}, respectively. 
The call of $\Split(Y',Y,r)$ in the statement~\ref{step: JHSKZ 10} of Algorithm~\ref{algCompCBt} is aimed to get the effects of the statements~\ref{step: UIIJW 6} and~\ref{step: UIIJW 7} of Algorithm~\ref{algCompCB}. The definition of this procedure, given under Algorithm~\ref{algCompCBt}, calls four subroutines in subsequent steps, which are discussed and defined below. 

\begin{figure*}
\begin{algorithm}[H]
\caption{\CompCBt\label{algCompCBt}}
\Input{a finite fuzzy graph $G = \tuple{V, E, L, \SV, \SE}$.}
\Output{the partition corresponding to the largest bisimulation of~$G$.}

\BlankLine
$\Initialize(G)$\;
\lIf{$\bbP.\Size() = 1$\label{step: JHSKZ 2}}{\Return $\bbP$} 
$\changed := \true$\;
\While{$\changed$}{
	$\changed := \false$\;
	\ForEach{$r \in \SE$}{
		\While{$\Not \bbQ[r].\compoundSuperBlocks.\Empty()$\label{step: JHSKZ 7}}{
			$Y := \bbQ[r].\compoundSuperBlocks.\first()$\;
			$Y' := Y.\smallerBlock()$\; 
			$\Split(Y',Y,r)$\label{step: JHSKZ 10}\tcp*{defined below}
			$\changed := \true$\;
		}
	}
}
\Return $\bbP$\label{step: JHSKZ 12}\; 
\end{algorithm}

\begin{procedure}[H]
	\caption{Split($Y'\!:\!\Block,\ Y\!:\!\SuperBlock,\ r\!:\!\SE$)\label{proc}}
	\red{create a vector $\verticesYp$ consisting of all the vertices of $Y'$\;}
	$\ComputeBlockEdges(\red{\verticesYp},r)$\;
	$\ComputeSubblocks(\red{\verticesYp},r)$\;
	$\DoSplitting(Y',\red{\verticesYp,} Y,r)$\;
	$\ClearAuxiliaryInfo(\red{\verticesYp},r)$\;
\end{procedure}
\end{figure*}

The call $\ComputeBlockEdges(\red{\verticesYp},r)$ in the procedure $\Split(Y',Y,r)$ prepares block-edges that will connect vertices (via normal edges labeled by~$r$) to the two future super-blocks of the super-partition $\bbQ[r]$, which have contents $Y\setminus Y'$ or $Y'$, respectively. Each block-edge that connects a vertex $x$ to the super-block $Y$ via normal edges labeled by~$r$ is updated to become the one that
\begin{itemize}
\item plays the role of a block-edge connecting $x$ to the future super-block with contents $Y \setminus Y'$ of $\bbQ[r]$, and 
\item has the attribute $\departingBlockEdge$ set to a newly created block-edge intended for connecting $x$ to the future super-block with contents $Y'$ of $\bbQ[r]$. 
\end{itemize}
The procedure $\ComputeBlockEdges(\red{\verticesYp},r)$ is formally defined on page~\pageref{proc: ComputeBlockEdges}. It uses the statement {\em continue} with the same semantics as in the C or Java language. To facilitate understanding this procedure, the reader may recall the description and specification of the class $\BlockEdge$.

\begin{figure*}
\begin{procedure}[H]
\caption{ComputeBlockEdges($\red{\verticesYp},r$)\label{proc: ComputeBlockEdges}}
\ForEach{$y \in \red{\verticesYp}$ and $e \in y.\comingEdges\red{[r]}$}{
	$bE := e.\blockEdge$\;
	\If{$bE.\departingBlockEdge = \Null$}{
		$bE.\departingBlockEdge = \New \BlockEdge(bE)$\;
	}
	$dbE = bE.\departingBlockEdge$\;
	$bE.\popKey(e.\degree)$, $dbE.\pushKey(e.\degree)$\;
}
\end{procedure}

\smallskip

\begin{procedure}[H]
	\caption{ComputeSubblocks($\red{\verticesYp},r$)\label{proc: ComputeSubblocks}}
	\ForEach{$y \in \red{\verticesYp}$ and $e \in y.\comingEdges\red{[r]}$}{
		$x := e.\edgeOrigin$, $bx := x.\block$\; 
		\lIf{$x.\processed$}{\Continue}
		$bE := e.\blockEdge$, $dbE := bE.\departingBlockEdge$\;
		$d_1 := bE.\maxKey()$, $d_2 := dbE.\maxKey()$\; 
		\uIf{$d_1 \geq d_2$}{
			\If{$d_2 \notin bx.\departingBlocks_2.\keys()$}{
				$bx.\departingBlocks_2[d_2] := \New \VertexList()$\;
			}
			\red{move $x$ from $bx.\vertices$ to $bx.\departingBlocks_2[d_2]$\;}
		}
		\Else{
			\If{$d_1 \notin bx.\departingBlocks_1.\keys()$}{
				$bx.\departingBlocks_1[d_1] := \New \VertexList()$\;
			}
			\red{move $x$ from $bx.\vertices$ to $bx.\departingBlocks_1[d_1]$\;}
		}
		$x.\processed := \true$\;
	}
\end{procedure}
\end{figure*}

\begin{figure*}
\begin{procedure}[H]
	\caption{DoSplitting($Y',\red{\verticesYp,} Y,r$)\label{proc: DoSplitting}}
	$Y.\removeBlock(Y')$\;
	$\createSuperBlock(Y.\superPartition, Y', r)$\;
	
	\BlankLine
	
	\ForEach{$y \in \red{\verticesYp}$ and $e \in y.\comingEdges\red{[r]}$}{
		\If{\markRed $e.\blockEdge.\departingBlockEdge \neq \Null$}{
			$e.\blockEdge := e.\blockEdge.\departingBlockEdge$\;
		} 
		$x := e.\edgeOrigin$, $bx := x.\block$\; 
		\If{$\Not (bx.\departingBlocks_1.\Empty()$ and $bx.\departingBlocks_2.\Empty())$}{
			{\markRed
			\If{$bx.\vertices.\Empty()$\label{step: HGPEB-0}}{
				\uIf{$\Not bx.\departingBlocks_1.\Empty()$}{
					let $d$ be the first key of the map $bx.\departingBlocks_1$\;
					$bx.\vertices := bx.\departingBlocks_1[d]$\;
					remove the key $d$ from the map $bx.\departingBlocks_1$\;
				}
				\Else {
					let $d$ be the first key of the map $bx.\departingBlocks_2$\;
					$bx.\vertices := bx.\departingBlocks_2[d]$\;
					remove the key $d$ from the map $bx.\departingBlocks_2$\label{step: HGPEB-1}\;
				}
			}
			}
			\ForEach{$d \in bx.\departingBlocks_1.\keys()$}{
				$\createBlock(bx.\departingBlocks_1[d])$\label{step: HGFDZ 9}\; 
			}
			\ForEach{$d \in bx.\departingBlocks_2.\keys()$}{
				$\createBlock(bx.\departingBlocks_2[d])$\label{step: HGFDZ 11}\;
			} 
			$bx.\departingBlocks_1.\clear()$\label{step: HGFDZ 12}\;
			$bx.\departingBlocks_2.\clear()$\label{step: HGFDZ 13}\; 
		}
	}
\end{procedure}

\smallskip

\begin{procedure}[H]
	\caption{ClearAuxiliaryInfo($\red{\verticesYp},r$)\label{proc: ClearAuxiliaryInfo}}
	
	\ForEach{$y \in \red{\verticesYp}$ and $e \in y.\comingEdges\red{[r]}$}{
		$e.\edgeOrigin.\processed := \false$\; 
		$bE := e.\blockEdge$, $sbE := bE.\sourceBlockEdge$\;
		\If{$sbE\neq \Null$}{
			$sbE.\departingBlockEdge := \Null$\;
			$bE.\sourceBlockEdge := \Null$\;
		}
	}
\end{procedure}
\end{figure*}


The call $\ComputeSubblocks(\red{\verticesYp},r)$ in the second statement of the procedure $\Split(Y',Y,r)$ computes for each related block $bx$ of $\bbP$ the attributes $\departingBlocks_1$ and $\departingBlocks_2$, which are the maps specified in the description of the class $\Block$. Let $X$ denote $bx.\vertices$. Then, recall that each key's value of the maps should identify a subblock resulting from splitting $X$ using $\tuple{Y',Y,r}$. The computation is done using the block-edges prepared by the call $\ComputeBlockEdges(\red{\verticesYp},r)$ discussed above.  Namely, for an arbitrary vertex $x \in X$, let $bE$ denote the block-edge connecting $x$ to $Y$ via normal edges labeled by $r$, and let $dbE = bE.\departingBlockEdge$, $d_1 = bE.\maxKey()$ and $d_2 = dbE.\maxKey()$. Then, the subblock containing $x$ is uniquely identified by the pair $(d_1,d_2)$. That is, two vertices of $X$ should be put into the same subblock iff they give the same pair $(d_1,d_2)$. As an invariant of the loop in the statement~\ref{step: JHSKZ 7} of the algorithm, $X$ is stable with respect to $\tuple{Y,r}$. Hence, $\max(d_1,d_2)$ is the same for all $x \in X$. Let $d = \max(d_1,d_2)$. If $d_1 \geq d_2$, then $d_1 = d$ and the subblock containing $x$ is uniquely identified by $d_2$ and this case (i.e., the truth of $d_1 \geq d_2$), so we put $x$ into $bx.\departingBlocks_2[d_2]$ (i.e., $bx.\departingBlocks_2[d_2]$ is used to represent that subblock). Dually, if $d_1 < d_2$, then we put $x$ into $bx.\departingBlocks_1[d_1]$. The procedure $\ComputeSubblocks(\red{\verticesYp},r)$ is formally defined on page~\pageref{proc: ComputeSubblocks}. 


The first two statements of the procedure $\Split(Y',Y,r)$ only prepare the block-edges and subblocks to be created for the splitting. The splitting itself is done by the procedure $\DoSplitting(Y',\red{\verticesYp,} Y,r)$ defined on page~\pageref{proc: DoSplitting}. 
\red{If the attribute $\vertices$ of a block $bx$ is an empty list, then the statements~\eqref{step: HGPEB-0}--\eqref{step: HGPEB-1} leave one of the subblocks of $bx$ in $bx$ (so that we do not have to delete $bx$ after the splitting).}


The attribute $\processed$ is used as an auxiliary Boolean flag for vertices in the processing of $\ComputeSubblocks(\red{\verticesYp},r)$ in order to avoid redundant computations. The attributes $\departingBlockEdge$ and $\sourceBlockEdge$ of block-edges as well as the attributes $\departingBlocks_1$ and $\departingBlocks_2$ of blocks are also used as auxiliary data during the processing of $\Split(Y',Y,r)$. The two latter attributes are cleared by the procedure $\DoSplitting(Y',\red{\verticesYp,} Y,r)$ itself. The remaining auxiliary data are cleared by 
the procedure $\ClearAuxiliaryInfo(\red{\verticesYp},r)$ (defined on page~\pageref{proc: ClearAuxiliaryInfo}). 

As discussed in this subsection, Algorithm~\ref{algCompCBt} reflects the run of Algorithm~\ref{algCompCB}. Hence, we arrive at the following counterpart of Theorem~\ref{theorem: UDFND}. 

\begin{theorem}
Algorithm~\ref{algCompCBt} is a correct algorithm for computing the partition corresponding to the largest bisimulation of a given finite fuzzy graph.
\end{theorem}

\subsection{Complexity Analysis}
\label{sec: comp-anal}

Recall that $n = |V|$, \red{$m = |E|$} and $l = |\{E(x,r,y)$ : $\tuple{x,r,y} \in$ \mbox{$V \times \SE \times V\}|$}. Assume that $l \geq 2$. 
Also recall that the sizes of $\SE$ and $\SV$ are assumed to be bounded by a constant. 
We now estimate the time complexity of Algorithm~\ref{algCompCBt} in terms of $n$, $m$ and $l$. 

Given \mbox{$Y':\Block$}, we write $|Y'|$ and $|\!\uparrow_r\!Y'|$ to denote the number of \red{all} vertices of $Y'$ and the number of edges labeled by $r$ and coming to the vertices of $Y'$. Given \mbox{$Y:\SuperBlock$}, we also write $|Y|$ to denote the number of vertices of $Y$. 

Observe that each iteration of the loop in the procedure $\ComputeBlockEdges(\red{\verticesYp},r)$ or $\ComputeSubblocks(\red{\verticesYp},r)$ takes time of order $O(\log{l})$. Thus, the complexities of these procedures are of order $O(|Y'| + |\!\uparrow_r\!Y'|\cdot \log{l})$. 

Consider a call of the procedure $\DoSplitting(Y',\red{\verticesYp},Y,r)$. The total time the statements~\ref{step: HGFDZ 9} and~\ref{step: HGFDZ 11} of this procedure take is of order $O(|\!\uparrow_r\!Y'|)$, because the created blocks are pairwise disjoint and the size of their union is bounded by $|\!\uparrow_r\!Y'|$. \red{For a similar reason, the total time the statements~\ref{step: HGPEB-0}-\ref{step: HGPEB-1}, \ref{step: HGFDZ 12} and \ref{step: HGFDZ 13} take is also of order $O(|\!\uparrow_r\!Y'|)$.} Thus, the complexity of this procedure is of order $O(|Y'| + |\!\uparrow_r\!Y'|)$. 

Clearly, the procedure $\ClearAuxiliaryInfo(\red{\verticesYp},r)$ also takes time of order $O(|Y'| + |\!\uparrow_r\!Y'|)$. 
Summing up, the complexity of the procedure $\Split(Y',Y,r)$ is of order $O(|Y'| + |\!\uparrow_r\!Y'|\cdot \log{l})$. Dividing this cost for \red{the} individual vertices of $Y'$, we can assume that the cost assigned to each vertex $y$ of $Y'$ in a call $\Split(Y',Y,r)$ is of order $O(1 + |r^{-1}(y)|\cdot \log{l})$, where $r^{-1}(y) = \{x \in V \mid$ $E(x,r,y) > 0\}$. 

Fix arbitrary $y \in V$ and $r \in \SE$. Let's estimate the number of calls $\Split(Y',Y,r)$ during the execution of Algorithm~\ref{algCompCBt} for $G$ such that $y$ is a vertex of $Y'$. Denote it by $f(y,r)$. Observe that, if $\Split(Y',Y,r)$ is such a call at some step, then the next call of $\Split$ with that property at some later step, if it exists, must be $\Split(U',U,r)$ with $U$ being a super-block whose set of vertices is a subset of the set of vertices of $Y'$. We have $|U'| \leq |U|/2$ and $|U| \leq |Y'| \leq |Y|/2$. Hence, $|U'| \leq |Y'|/2$. Extending this understanding, we conclude that $f(y,r) \leq \log{n}$. 

Therefore, the total time taken by all the calls of $\Split$ in the statement~\ref{step: JHSKZ 10} of the execution of Algorithm~\ref{algCompCBt} for $G$ is of order 
\[ O(
	\sum_{y \in V} \sum_{r \in \SE} \log{n} \cdot (1 + |r^{-1}(y)|\cdot \log{l}) 
), \]
which is of order $O((m\log{l} + n)\log{n})$. 
As estimated in Section~\ref{sec: init}, the call $\Initialize(G)$ takes time of order \mbox{$O((m+n)\log{n})$}. Hence, we arrive at the following theorem. 

\begin{theorem}
Algorithm~\ref{algCompCBt} has a time complexity of order $O((m\log{l} + n)\log{n})$.
\end{theorem}

If $l$ is bounded by a constant (e.g., when $l = 2$ and $G$ is a crisp graph), then the time complexity of Algorithm~\ref{algCompCBt} is of order $O((m+n)\log{n})$. If \mbox{$m \geq n$}, then, taking $l = n^2$ for the \modified{worst} case, the complexity of the algorithm is of order $O(m\log^2(n))$. 

\section{Crisp Bisimulations with Counting Successors}
\label{sec: setting-counting-successors}

In this section, we study the problem of computing crisp bisimulations for fuzzy graphs under the setting with counting successors. In the following, let $G = \tuple{V, E, L, \SV, \SE}$ be a finite fuzzy graph. 
If not stated otherwise, we still use $\bbP$, $\bbQ$ and $\bbS$ to denote partitions of $V$, $X$ and $Y$ to denote non-empty subsets of $V$, and $r$ to denote an edge label from $\SE$. 

We write $x\!\uparrow_r$ to denote the set $\{y \in V \mid E(x,r,y) > 0\}$. 

A non-empty binary relation $Z \subseteq V \times V$ is called a {\em crisp auto-bisimulation of $G$ with counting successors}, or a {\em s-bisimulation} of $G$ for short, if Condition~\eqref{eq: CB 1} and the following one hold for all $x,x' \in V$ and $r \in \SE$:
\begin{align}
& \textrm{if $Z(x,x')$ holds, then there exists a bijection $h: x\!\uparrow_r\ \to\ x'\!\uparrow_r$ such that,} \nonumber \\[-0.5ex]
& \textrm{for every $y \in x\!\uparrow_r$, $Z(y,h(y))$ holds and $E(x,r,y) = E(x',r,h(y))$.} \label{eq: xCB 4}
\end{align}

By using~\cite[Lemma~6.1]{BSDL-P-LOGCOM}, it can be shown that the above definition coincides with the one of~\cite[Section~4.1]{jBSfDL2} for the case when $\mI = \mI'$, $\mI$ is a finite fuzzy interpretation, and $\Phi = \{Q_n \mid n \in \NN \setminus \{0\}\}$.  

The objective is to develop an efficient algorithm for computing the largest s-bisimulation of a given finite fuzzy graph.

A counterpart of Proposition~\ref{prop: MNGOA} in which $G$ is finite and ``bisimulation(s) of $G$'' is replaced by ``s-bisimulation(s) of $G$'' also holds and can easily be proved. 
As a consequence, we obtain the following counterpart of Corollary~\ref{cor: JHDKS}.

\begin{corollary}\label{cor: JHDKS2}
The largest s-bisimulation of a finite fuzzy graph exists and is an equivalence relation.
\end{corollary}

\newcommand{\xrpY}[4]{#1\!\uparrow_{#2(#3)}\!#4}
\newcommand{\xrdY}[1]{\xrpY{#1}{r}{d}{Y}}

We define $\xrdY{x} = \{y \in Y \mid E(x,r,y) = d\}$. 

We say that $X$ is {\em s-stable with respect to $\tuple{Y,r}$} (and $G$) if 
\mbox{$|\xrdY{x}| = |\xrdY{x'}|$} for all $x,x' \in X$ and $d \in (0,1]$. 

We say that $\bbP$ is {\em s-stable with respect to $\tuple{Y,r}$} (and $G$) if every $X \in \bbP$ is s-stable with respect to $\tuple{Y,r}$. In addition, $\bbP$ is {\em s-stable} (with respect to $G$) if it is s-stable with respect to $\tuple{Y,r}$ for all $Y \in \bbP$ and $r \in \SE$. 

By $\bbP_1$ we denote the partition of $V$ that corresponds to the equivalence relation 
\[ \textrm{$\{\tuple{x,x'} \in V^2 \mid$ $L(x) = L(x')$ and $|\xrpY{x}{r}{d}{V}| = |\xrpY{x'}{r}{d}{V}|$ for all $r \in \SE$ and $d \in (0,1]\}$.} \] 

The following lemma is a counterpart of Lemma~\ref{lemma: HGRLA}, which allows to look at the problem from another point of view. 

\begin{lemma}\label{lemma: HGRLA2}
$\bbQ$ is the partition corresponding to the largest s-bisimulation of $G$ iff it is the coarsest s-stable refinement of~$\bbP_1$.  
\end{lemma}

\begin{proof}
It is sufficient to prove the following assertions:
\begin{enumerate}
	\item If a partition $\bbQ$ is an s-stable refinement of~$\bbP_1$, then its corresponding equivalence relation is an s-bisimulation of~$G$. 
	\item If $\bbQ$ is the partition corresponding to the largest s-bisimulation of $G$, then it is an s-stable refinement of~$\bbP_1$.  
\end{enumerate}

Consider the first assertion. 
Let $\bbQ$ be an s-stable refinement of~$\bbP_1$ and $Z$ the equivalence relation corresponding to $\bbQ$. We need to show that $Z$ satisfies Conditions \eqref{eq: CB 1} and~\eqref{eq: xCB 4}. Condition~\eqref{eq: CB 1} holds since $\bbQ$ is a refinement of $\bbP_1$. Consider Condition~\eqref{eq: xCB 4}. Let $x,x' \in V$, $r \in \SE$ and assume that $Z(x,x')$ holds. 
Thus, there exists $X \in \bbQ$ such that $x,x' \in X$. Furthermore, for every $Y \in \bbQ$ and $d \in (0,1]$, \mbox{$\xrdY{x} = \xrdY{x'}$}. For every $Y \in \bbQ$ and $d \in (0,1]$, let $h_{Y,d}$ be a bijection between $\xrdY{x}$ and $\xrdY{x'}$. Let \mbox{$h: x\!\uparrow_r\ \to\ x'\!\uparrow_r$} be the function specified as follows: for $y \in x\!\uparrow_r$, let $d = E(x,r,y)$ (we have that $d > 0$) and let $Y$ be the block of $\bbQ$ such that $y \in Y$, then define $h(y) = h_{Y,d}(y)$. It is easy to see that $h$ is a bijection and, for every $y \in x\!\uparrow_r$, $Z(y,h(y))$ holds and $E(x,r,y) = E(x',r,h(y))$. 
This completes the proof of Condition~\eqref{eq: xCB 4}. 

Consider the second assertion. 
Let $\bbQ$ be the partition corresponding to the largest s-bisimulation $Z$ of $G$. Due to Conditions~\eqref{eq: CB 1} and~\eqref{eq: xCB 4}, $\bbQ$ is a refinement of $\bbP_1$. 
It remains to show that $\bbQ$ is s-stable. Let $X,Y \in \bbQ$, $x,x'\in X$, $r \in \SE$ and $d \in (0,1]$. We need to show that \mbox{$\xrdY{x} = \xrdY{x'}$}.
By definition, then there exists a bijection $h: x\!\uparrow_r\ \to\ x'\!\uparrow_r$ such that, for every $y \in x\!\uparrow_r$, $Z(y,h(y))$ holds and $E(x,r,y) = E(x',r,h(y))$. Thus, \mbox{$h(\xrdY{x}) = \xrdY{x'}$}, and consequently, \mbox{$\xrdY{x} = \xrdY{x'}$}. This completes the proof. 
\myend
\end{proof}

Let $\emptyset \subset Y' \subset Y$. 
By $\sSplit(X,\tuple{Y',Y,r})$ we denote the coarsest partition of $X$ such that each of its blocks is s-stable with respect to both $\tuple{Y',r}$ and $\tuple{Y\setminus Y',r}$. Clearly, that partition exists and is computable. We also define 
\[ \sSplit(\bbP,\tuple{Y',Y,r}) = \bigcup \{\sSplit(X,\tuple{Y',Y,r}) \mid X \in \bbP\}. \]
It can be seen that $\sSplit(\bbP,\tuple{Y',Y,r})$ is the coarsest refinement of $\bbP$ that is s-stable with respect to both $\tuple{Y',r}$ and $\tuple{Y\setminus Y',r}$. 

Let \sCompCB be the algorithm obtained from Algorithm~\ref{algCompCB} by replacing $\bbP_0$ and $\itSplit$ with  $\bbP_1$ and $\sSplit$, respectively. 
An example illustrating this algorithm is given below.

\begin{example}\label{example: ill-2}
	Consider the execution of the algorithm \sCompCB for the fuzzy graph $G$ specified in Example~\ref{example: ill-1}. 
	\begin{itemize}
		\item At the beginning, we have that 
		$\bbP = \bbP_1 = \{\{a\}$, $\{b\}$, $\{c,f,g\}$, $\{d,e\}\}$ and $\bbQ_r = \{V\}$. 
		
		\item Assume that $Y' = \{a\}$ in the first iteration of the main loop. 
		As the effects of this iteration, $\bbP$ is refined to $\{\{a\}, \{b\}, \{c,f\}, \{d,e\}, \{g\}\}$, and $\bbQ_r$ is changed to $\{\{a\}$, $\{b,c,d,e,f,g\}\}$. 
		\item In the next three iterations of the main loop, $\bbP$ does not change any more, but $\bbQ_r$ is refined gradually until becoming $\bbP$. The algorithm terminates with the current $\bbP$ as the result.
		\myend
	\end{itemize}
\end{example}

The following theorem is a counterpart of Theorem~\ref{theorem: UDFND}. It can be proved analogously.

\begin{theorem}\label{theorem: UDFND 2}
\sCompCB is a correct algorithm for computing the partition corresponding to the largest s-bisimulation of a given finite fuzzy graph.
\end{theorem}

Let's now consider the problem how to implement the algorithm \sCompCB to obtain an efficient algorithm by revising Algorithm~\ref{algCompCBt}. 

We define that two objects of the class $\BlockEdge$ are {\em equal} (as in the Java language) if the maps specified by them are equal. This definition is essential, for example, when block-edges are used as keys for (higher-order) maps.

Let $\sInitialize$ be the procedure obtained from $\Initialize$ by replacing the condition $\blockEdges[x,r].\maxKey() = \blockEdges[x',r].\maxKey()$ in the statement~\ref{step: HGJWA 17} with the condition that $\blockEdges[x,r]$ and $\blockEdges[x',r]$ are equal. The reason is that the algorithm \sCompCB uses $\bbP_1$ instead of $\bbP_0$.

Observe that, as an invariant of the loop of the algorithm \sCompCB, $\bbP$ is s-stable with respect to $\tuple{Y,r}$ for all $r \in \SE$ and $Y \in \bbQ_r$. 
Furthermore, if $X$ is s-stable with respect to $\tuple{Y,r}$ and $0 \subset Y' \subset Y$, then $X$ is s-stable with respect to $\tuple{Y',r}$ iff it is s-stable with respect to $\tuple{Y\setminus Y',r}$. 
Thus, the algorithm \sCompCB remains the same if we modify the definition of $\sSplit$ so that $\sSplit(X,\tuple{Y',Y,r})$ is the coarsest partition of $X$ such that each of its blocks is s-stable with respect to $\tuple{Y\setminus Y',r}$. 

\begin{figure*}
	\begin{procedure}[H]
		\caption{s-ComputeSubblocks($\red{\verticesYp},r$)\label{proc: sComputeSubblocks}}
		\ForEach{$y \in \red{\verticesYp}$ and $e \in y.\comingEdges\red{[r]}$}{
			$x := e.\edgeOrigin$, $bx := x.\block$\; 
			\If{not $x.\processed$}{
				$bE := e.\blockEdge$\;
				
				\If{$bE \notin bx.\departingBlocks.\keys()$}{
					$bx.\departingBlocks[bE] := \New \VertexList()$\;
				}
				
				\red{move $x$ from $bx.\vertices$ to $bx.\departingBlocks[bE]$\label{step: JHDJQ 6}\;}
				$x.\processed := \true$\;
			}
		}
	\end{procedure}
	
	\smallskip
	
	\begin{procedure}[H]
		\caption{s-DoSplitting($Y',\red{\verticesYp,} Y,r$)\label{proc: sDoSplitting}}
		$Y.\removeBlock(Y')$\;
		$\createSuperBlock(Y.\superPartition, Y', r)$\;
		
		\BlankLine
		
		\ForEach{$y \in \red{\verticesYp}$ and $e \in y.\comingEdges\red{[r]}$}{
			$e.\blockEdge := e.\blockEdge.\departingBlockEdge$\; 
			$x := e.\edgeOrigin$, $bx := x.\block$\; 
			
			\If{\markRed $\Not bx.\departingBlocks.\Empty()$}{
				{\markRed
				\If{$bx.\vertices.\Empty()$\label{step: HGPEB-2}}{
					let $bE$ be the first key of the map $bx.\departingBlocks$\;
					$bx.\vertices := bx.\departingBlocks[bE]$\;
					remove the key $bE$ from the map $bx.\departingBlocks$\;
				}
				}		
				
				\ForEach{$bE \in bx.\departingBlocks.\keys()$}{
					$\createBlock(bx.\departingBlocks[bE])$\;
				} 
				$bx.\departingBlocks.\clear()$\;
			}
		}
	\end{procedure}
\end{figure*}

To reflect the above observation, we assume the following modifications for the class $\Block$ specified in Section~\ref{sec: DS}: 
\begin{itemize}
	\item the attributes \mbox{$\departingBlocks_i: [0,1] \to \VertexList$} 
	(for $i \in \{1,2\}$) are replaced by the attribute 
	\mbox{$\departingBlocks: \BlockEdge \to \VertexList$}; 
	\item the constructor sets $\departingBlocks$ instead of $\departingBlocks_i$ ($i \in \{1,2\}$) to a newly created empty map of the corresponding type. 
\end{itemize}
Then, instead of the procedure $\ComputeSubblocks(\red{\verticesYp},r)$ we will use its revision $\sComputeSubblocks(\red{\verticesYp},r)$ defined on page~\pageref{proc: sComputeSubblocks}. 
To explain this new procedure, consider its call instead of $\ComputeSubblocks(\red{\verticesYp},r)$ in the procedure $\Split(Y',Y,r)$. Recall that the block-edge that connects a vertex $x$ to the super-block $Y$ via normal edges labeled by~$r$, which is $bE$ in the definition of $\sComputeSubblocks$, has been updated by the call $\ComputeBlockEdges(\red{\verticesYp},r)$ to become the one that plays the role of a block-edge connecting $x$ to the future super-block with contents $Y \setminus Y'$ of $\bbQ[r]$. The subblock containing $x$ that should result from splitting the block of $x$ by $\tuple{Y',Y,r}$ is identified by $bE$. Therefore, we put $x$ to that subblock by the statement~\ref{step: JHDJQ 6} in the procedure $\sComputeSubblocks$. 
With these changes, it is natural to revise the procedure $\DoSplitting$ to get the procedure $\sDoSplitting$ shown on page~\pageref{proc: sDoSplitting}. 

Let $\sSfSplit$ be the procedure obtained from the procedure $\Split$ by replacing $\ComputeSubblocks$ and $\DoSplitting$ with $\sComputeSubblocks$ and $\sDoSplitting$, respectively. 
Then, let \sCompCBt be the algorithm obtained from Algorithm~\ref{algCompCBt} by replacing $\Initialize$ and $\Split$ with $\sInitialize$ and $\sSfSplit$, respectively. 
It can be seen that this revised algorithm reflects the run of the algorithm \sCompCB. Hence, we arrive at the following result. 

\begin{theorem}
\sCompCBt is a correct algorithm for computing the partition corresponding to the largest s-bisimulation of a given finite fuzzy graph.
\end{theorem}

Observe that the number of possible keys of the attribute $\departingBlocks$ of objects of the revised class $\Block$ is bounded by $m$. Consequently, the complexity of the procedure $\sSfSplit(Y',Y,r)$ is of order $O(|Y'| + |\!\uparrow_r\!Y'|\cdot \log{m})$. Using the technique applied in Section~\ref{sec: comp-anal}, we arrive at the following theorem. 

\begin{theorem}
The time complexity of the algorithm \sCompCBt is of order $O((m\log{m}$ + $n)\log{n})$.
\end{theorem}

Note that, if $m \geq n$, then the complexity order in the above theorem can be simplified to $O(m\log^2(n))$. 

{\markRed
\section{Experiments}
\label{sec: experiments}

\newcommand{\crispbis}{\mbox{\textit{crispbis}}\xspace}
\newcommand{\crispbisCS}{\mbox{\textit{crispbis\_cs}}\xspace}
\newcommand{\crispbisPY}{\mbox{\textit{CompCB.py}}\xspace}
\newcommand{\genTest}{\mbox{\textit{genTest}}\xspace}
\newcommand{\genTestA}{\mbox{\textit{genTest1}}\xspace}
\newcommand{\genTestB}{\mbox{\textit{genTest2}}\xspace}
\newcommand{\genTestC}{\mbox{\textit{genTest3}}\xspace}
\newcommand{\withCountingSuccessors}{\mbox{\textit{-$\,$-withCountingSuccessors}}\xspace}
\newcommand{\naive}{\mbox{\textit{-$\,$-naive}}\xspace}

We have implemented both the algorithms \CompCBt and \sCompCBt in Python and C++~\cite{compCBfFS-code}, resulting in the programs \crispbis and \crispbisCS, compiled from C++, and the program \crispbisPY in Python, where \crispbis and \crispbisPY implement the algorithm \CompCBt, whereas \crispbisCS and ``\crispbisPY \withCountingSuccessors'' (i.e., \crispbisPY with the given option) implement the algorithm \sCompCBt. For testing the correctness of the two mentioned algorithms, we have also implemented in Python naive algorithms for the considered problems, which can be run by using the same program \crispbisPY with the (additional) option $\naive$. All of the mentioned programs read the input fuzzy graph from the standard input stream and write the resulting partition to the standard output stream. For the format of inputs and outputs, we refer the reader to~\cite{compCBfFS-code}.  

The above mentioned naive counterpart of \CompCBt (respectively, \sCompCBt) is as follows. It keeps a doubly linked list of blocks of the current partition $\bbP$ (of $V$), which is initialized by using $\bbP_0$ (respectively, $\bbP_1$). While the flag $\changed$ is set on (as initially), if there exist blocks $X$ and $Y$ of $\bbP$ and an edge label $r \in \SE$ such that $X$ is not stable (respectively, s-stable) with respect to $\tuple{Y,r}$, then the algorithm replaces $X$ in $\bbP$ with the blocks obtained from splitting $X$ using $\tuple{Y,s}$, for all $s \in \SE$, so that all of them are stable (respectively, s-stable) with respect to $\tuple{Y,s}$, else it sets off the flag $\changed$ (to break the main loop). The implementation has been optimized by the use of the doubly linked list of blocks of the current partition and that the blocks obtained from splitting $X$ are added to the beginning of the list to better propagate the changes. Analyzing the complexity of this naive algorithm is not an easy task. In the worst case, the complexity is at least $\Omega((m+n)n)$, where $m = |E|$ and $n = |V|$.

\subsection{Generating Test Data}

We have implemented (in C++) a program \genTest for generating random tests for the above mentioned programs (\crispbis,  \crispbisCS and \crispbisPY). This program takes a number of parameters given as natural numbers, where the first one is 1, 2 or 3 and specifies which function, \genTestA, \genTestB or \genTestC, is called with the remaining parameters. These functions are specified below. All of them write (the information of) a randomly generated fuzzy graph to the standard output stream. 

\begin{table}
\markRed
\footnotesize
\begin{center}
\begin{tabular}{|c|l|r|r|r|r|r|r|}
\hline
Test & \multicolumn{1}{|c|}{Parameters for} & \multicolumn{1}{|c|}{$|V|$} & \multicolumn{1}{|c|}{$|E|$} & \multicolumn{1}{|c|}{Blocks} & \multicolumn{3}{|c|}{Time (ms)} \\
\cline{6-8} 
& \multicolumn{1}{|c|}{\genTest} & & & & naive Python & Python & C++ \\
\hline
1 & 1 100 0 0 & 20000 & 3960000 & 200 & 215730 & 97402 & 25689 \\
2 & 1 100 1 0 & 20000 & 4000000 & 200 & 214811 & 96218 & 25720 \\
3 & 2 100 0	  & 20000 & 2000000 & 100	& 59991 & 44947 & 13645 \\
\hline
4 & 3 100 10 20 0 3 0 2	& 1000 & 2000 & 852 & 64171 & 140 & 19 \\
5 & 3 100 10 20 10 10 1 2 & 1000 & 2000 & 890 & 33425	& 142 & 21 \\
6 & 3 100 10 50 0 3 0 1	& 1000 & 5000 & 804 & 53887 & 174 & 30 \\
7 & 3 100 10 50 5 5 1 1	& 1000 & 5000 & 1000 & 41159 & 193 & 26 \\
8 & 3 100 10 60 0 3 0 2	& 1000 & 6000 & 1000 & 57283 & 239 & 34 \\
9 & 3 100 10 60 10 10 1 2 & 1000 & 6000 & 1000 & 41917 & 228 & 35 \\
\hline
10 & 3 $10^5$ 10 20 0 3 0 2 & $10^6$ & $2*10^6$ & 791236 & & & 21783 \\
11 & 3 $10^5$ 10 20 10 10 1 2 & $10^6$ & $2*10^6$ & 841240 & & & 23651 \\
12 & 3 $10^5$ 10 50 0 3 0 1 & $10^6$ & $5*10^6$ & 833080 & & & 28274 \\
13 & 3 $10^5$ 10 50 5 5 1 1 & $10^6$ & $5*10^6$ & 998672 & & & 28868 \\
14 & 3 $10^5$ 10 60 0 3 0 2 & $10^6$ & $6*10^6$ & 998816 & & & 36730 \\
15 & 3 $10^5$ 10 60 10 10 1 2 & $10^6$ & $6*10^6$ & 999191 & & & 37821 \\
\hline
16 & 3 $10^5$ 10 60 10 $10^6$ 1 2 & $10^6$ & $6*10^6$ & $10^6$ & & & 43868 \\
17 & 3 1 $10^6$ ($6*10^6$) ($5*10^5$) 3 1 2 & $10^6$ & $6*10^6$ & 997487 & & & 61564 \\
\hline
\end{tabular}

\medskip

\begin{tabular}{|c|l|r|r|r|r|r|r|}
\hline
Test & \multicolumn{1}{|c|}{Parameters for} & \multicolumn{1}{|c|}{$|V|$} & \multicolumn{1}{|c|}{$|E|$} & \multicolumn{1}{|c|}{Blocks} & \multicolumn{3}{|c|}{Time (ms) for the setting} \\
& \multicolumn{1}{|c|}{\genTest} & & & & \multicolumn{3}{|c|}{with {\em counting successors}} \\
\cline{6-8} 
& & & & & naive Python & Python & C++ \\
\hline
18 & 1 100 0 1 & 20000 & 3960000 & 200 & 232607 & 91168 & 22712 \\
19 & 1 100 1 1 & 20000 & 4000000 & 200 & 236204 & 91817 & 22925 \\
20 & 2 100 1	& 20000 & 2000000 & 100 & 67237 & 43449 & 11198 \\
\hline
21 & 3 100 10 20 0 3 0 2 & 1000 & 2000 & 835 & 103839 & 139 & 20 \\
22 & 3 100 10 20 10 3 1 2 & 1000 & 2000 & 868 & 97271 & 137 & 20 \\
23 & 3 100 10 50 0 3 0 1 & 1000 & 5000 & 1000 & 152001 & 190 & 31 \\
24 & 3 100 10 50 5 3 1 1 & 1000 & 5000 & 1000 & 50382 & 189 & 27 \\
25 & 3 100 10 60 0 3 0 2 & 1000 & 6000 & 999 & 698646 & 225 & 41 \\
26 & 3 100 10 60 10 3 1 2 & 1000 & 6000 & 1000	& 104259 & 221 & 43 \\
\hline
27 & 3 $10^5$ 10 20 0 3 0 2 & $10^6$ & $2*10^6$ & 792680 & & & 21845 \\
28 & 3 $10^5$ 10 20 10 3 1 2 & $10^6$ & $2*10^6$ & 812089 & & & 24752 \\
29 & 3 $10^5$ 10 50 0 3 0 1 & $10^6$ & $5*10^6$ & 998609 & & & 26938 \\
30 & 3 $10^5$ 10 50 5 3 1 1 & $10^6$ & $5*10^6$ & 998656 & & & 28849 \\
31 & 3 $10^5$ 10 60 0 3 0 2 & $10^6$ & $6*10^6$ & 998896 & & & 35952 \\
32 & 3 $10^5$ 10 60 10 3 1 2 & $10^6$ & $6*10^6$ & 998907 & & & 39690 \\
\hline
33 & 3 $10^5$ 10 60 10 $10^6$ 1 2 & $10^6$ & $6*10^6$ & $10^6$ & & & 51009 \\
34 & 3 1 $10^6$ ($6*10^6$) ($5*10^5$) 3 1 2 & $10^6$ & $6*10^6$ & 997574 & & & 61862 \\
\hline
\end{tabular}
\end{center}
\caption{Execution times of the programs \crispbis, \crispbisCS and \crispbisPY for the input data generated by the program \genTest using the parameters given in the second column. The upper part contains the results of \crispbisPY without the option \withCountingSuccessors and \crispbis, whereas the lower part contains the results of \crispbisPY with the option \withCountingSuccessors and \crispbisCS. The execution times of \crispbisPY without (respectively, with) the $\naive$ option are given in the sixth (respectively, seventh) column. The execution times of \crispbis and \crispbisCS are given in the eighth column. A value in the third (respectively, fourth) column is the number of vertices (respectively, non-zero edges) of the fuzzy graph generated by \genTest. A value in the fifth column is the (average) number of blocks of the resulting partition. The tests have been done using a laptop with Intel$^{\textrm{\textregistered}}$ Core$^\mathrm{TM}$ i3-3120M CPU~@2.50GHz$\,\times\,$4 and 4~GB RAM. Each of them has been repeated three times. The results in the columns 5--8 are the averages from the three repeats. A blank field in the table mean we did not perform the corresponding test for the corresponding program.\label{table: JHDHG}}
\end{table}

\newcommand{\nxti}{i_{\oplus 1}}
\newcommand{\nxtj}{j_{\oplus 1}}

\medskip

\noindent {\bf Function \genTestA:} This function has parameters $k$, $cyclic$ and $countingSuccessors$, where the first one is a positive integer and the remaining ones are boolean values (0 or 1). It constructs a fuzzy graph $G = \tuple{V, E, L, \SV, \SE}$ as follows.
\begin{itemize}
\item $V := \{a_{i,j}, b_{i,j} \mid i,j \in \{0,\ldots,k-1\}\}$, $\SV := \{p,q\}$ and $\SE := \{r,s\}$;
\item let $d_1$, $d_2$, $d_{r,a}$, $d_{s,a}$, $d_{r,b}$, $d_{s,b}$ be random constants from $(0,1]$, where $d_1 \neq d_2$;
\item initialize $E$ and $L$ by setting them to the appropriate empty data structures;
\item for $i,j \in \{0,\ldots,k-1\}$:
	\begin{itemize}
	\item $L(a_{i,j})(p) := d_1$ and $L(b_{i,j})(p) := d_2$;
	\item if $cyclic$: $L(a_{0,j})(q) := 1$ and $L(b_{0,j})(q) := 1$;
	\item if $\lnot cyclic$ and $i = k-1$: break;
	\item let $\nxti$ denote $i+1$ if $i < k-1$ and $0$ otherwise;
	\item let $t_{aa}$, $t_{ab}$, $t_{ba}$ and $t_{bb}$ be random indices from $\{0,\ldots,k-1\}$;
	\item $E(a_{i,j}, r, a_{\nxti,t_{aa}}) := d_{r,a}$ and $E(a_{i,j}, s, a_{\nxti,t_{aa}}) := d_{s,a}$;
	\item $E(a_{i,j}, r, b_{\nxti,t_{ab}}) := d_{r,b}$ and $E(a_{i,j}, s, b_{\nxti,t_{ab}}) := d_{s,b}$;
	\item $E(b_{i,j}, r, a_{\nxti,t_{ba}}) := d_{r,a}$ and $E(b_{i,j}, s, a_{\nxti,t_{ba}}) := d_{s,a}$;
	\item $E(b_{i,j}, r, b_{\nxti,t_{bb}}) := d_{r,b}$ and $E(b_{i,j}, s, b_{\nxti,t_{bb}}) := d_{s,b}$;
	\item let $d'_{r,a}$, $d'_{s,a}$, $d'_{r,b}$ and $d'_{s,b}$ be random values from $(0,1]$ that are less than $d_{r,a}$, $d_{s,a}$, $d_{r,b}$ and $d_{s,b}$, respectively;
	\item for each $0 \leq j' < k$:
		\begin{itemize}
		\item if $\lnot countingSuccessors$: let $d'_{r,a}$, $d'_{s,a}$, $d'_{r,b}$ and $d'_{s,b}$ be random values from $(0,1]$ that are less than or equal to $d_{r,a}$, $d_{s,a}$, $d_{r,b}$ and $d_{s,b}$, respectively; 
		\item if $j' \neq t_{aa}$: $E(a_{i,j}, r, a_{\nxti,j'}) := d'_{r,a}$ and $E(a_{i,j}, s, a_{\nxti,j'}) := d'_{s,a}$; 
		\item if $j' \neq t_{ab}$: $E(a_{i,j}, r, b_{\nxti,j'}) := d'_{r,b}$ and $E(a_{i,j}, s, b_{\nxti,j'}) := d'_{s,b}$;
		\item if $j' \neq t_{ba}$: $E(b_{i,j}, r, a_{\nxti,j'}) := d'_{r,a}$ and $E(b_{i,j}, s, a_{\nxti,j'}) := d'_{s,a}$;
		\item if $j' \neq t_{bb}$: $E(b_{i,j}, r, b_{\nxti,j'}) := d'_{r,b}$ and $E(b_{i,j}, s, b_{\nxti,j'}) := d'_{s,b}$. 
		\end{itemize}
	\end{itemize}
\end{itemize}
We have $|V| = 2 k^2$ and $|E| \in \{4(k-1)k^2, 4k^3\}$ (depending on the boolean option $cyclic$).  
Roughly speaking, the generated fuzzy graph $G$ contains $k$ layers of vertices, where the \mbox{$i$-th} layer consists of $a_{i,j}$ and $b_{i,j}$ for all $0 \leq j < k$, and each layer is ``connected'' only to the next one. If the parameter $cyclic$ is $\true$, then the last layer is ``connected'' to the first one. The construction of $G$ is designed so that, if the parameter $countingSuccessors$ is $\false$ (respectively, $\true$), then running the algorithm \CompCBt (respectively, \sCompCBt) for $G$ should result in the partition consisting of the blocks $a_i = \{a_{i,j} \mid 0 \leq j < k\}$ and $b_i = \{b_{i,j} \mid 0 \leq j < k\}$, for $0 \leq i < k$. 

\medskip

\noindent \mbox{\bf Function \genTestB:} This function has parameters $k$ (a positive integer) and $countingSuccessors$ (a boolean value, 0 or 1). It constructs a fuzzy graph $G = \tuple{V, E, L, \SV, \SE}$ as follows. 
\begin{itemize}
\item $V := \{a_{i,j}, b_{i,j} \mid i,j \in \{0,\ldots,k-1\}\}$, $\SV := \{p\}$ and $\SE := \{r,s\}$;
\item initialize $E$ and $L$ by setting them to the appropriate empty data structures;
\item for $0 \leq i < k$: $L(a_{i,0})(p) := 1$ and $L(b_{i,0})(p) := 1$;
\item for $i,j \in \{0,\ldots,k-1\}$: 
	\begin{itemize}
	\item let $\nxti$ denote $i+1$ if $i < k-1$ and $0$ otherwise;
	\item let $\nxtj$ denote $j+1$ if $j < k-1$ and $0$ otherwise;
	\item for $0 \leq i' < k$: 
		\begin{itemize}
		\item if ($i' = i$ and $j < k - 1$) or ($i' = \nxti$ and $j = k - 1$):\\ set $E(a_{i,j}, r, a_{i',\nxtj})$, $E(a_{i,j}, s, a_{i',\nxtj})$, $E(b_{i,j}, r, b_{i',\nxtj})$ and $E(b_{i,j}, s, b_{i',\nxtj})$ to $1$;
		\item else if $\lnot countingSuccessors$: set each among $E(a_{i,j}, r, a_{i',\nxtj})$, $E(a_{i,j}, s, a_{i',\nxtj})$, $E(b_{i,j}, r, b_{i',\nxtj})$ and $E(b_{i,j}, s, b_{i',\nxtj})$ to a random value from $(0,1]$;
		\item else set $E(a_{i,j}, r, a_{i',\nxtj})$, $E(a_{i,j}, s, a_{i',\nxtj})$, $E(b_{i,j}, r, b_{i',\nxtj})$ and $E(b_{i,j}, s, b_{i',\nxtj})$ to~$1/k$.
		\end{itemize}
	\end{itemize}
\end{itemize}
We have $|V| = 2 k^2$ and $|E| = 2k^3$.  
Roughly speaking, the generated fuzzy graph $G$ contains two rings. The first ring consists of all $a_{i,j}$, for $i,j \in \{0,\ldots,k-1\}$, in which $a_{i,j}$ follows $a_{i,j-1}$ if $j > 0$ and follows $a_{i-1,k-1}$ if $j = 0$ and $i > 0$, whereas $a_{0,0}$ follows $a_{k-1,k-1}$. Each node of the ring is connected to the next one and $k-1$ other nodes of the ring. The second ring is similar and consists of all $b_{i,j}$, for $i,j \in \{0$, \ldots, $k-1\}$. The construction of $G$ is designed so that, if the parameter $countingSuccessors$ is $\false$ (respectively, $\true$), then running the algorithm \CompCBt (respectively, \sCompCBt) for $G$ should result in the partition consisting of the blocks $B_j = \{a_{i,j}, b_{i,j} \mid 0 \leq i < k\}$, for $0 \leq j < k$. 

\medskip

\noindent \mbox{\bf Function \genTestC:} This function constructs a random fuzzy graph $G = \tuple{V, E, L, \SV, \SE}$ using parameters $k$, $n'$, $m'$, $p$, $l$, $|\SV|$ and $|\SE|$, which are positive integers with the following meanings:
\begin{itemize}
\item $|\SV|$ and $|\SE|$ are the sizes of $\SV$ and $\SE$, respectively,
\item $l$ is the number of different non-zero fuzzy values used for $G$,
\item $k$ is the number of pairwise disjoint fuzzy subgraphs of $G$, each of which consists of $n'$ vertices, $m'$ non-zero edges and $p$ non-zero vertex labels (i.e., $p$ is the number of pairs $(v,q)$ such that $v$ is a vertex of the fuzzy subgraph, $q \in \SV$ and $L(v)(q) > 0$),
\end{itemize}
where random generations are done using uniform distributions. We have $|V| = k \times n'$ and $|E| = k \times m'$.

\subsection{Test Results}

We have performed many tests for the implemented programs. Their descriptions and results are reported in Table~\ref{table: JHDHG}. The aim of Tests 1--9 and 18--26 is to check the correctness of the algorithms and their implementations. All of them passed positively. Some notable observations are as follows:
\begin{itemize}
\item For Test~25, with $|V| = 1000$ and $|E| = 6000$, the naive version written in Python for the setting with counting successors runs about 3100 (respectively, 17000) times slower than the corresponding sophisticated version (i.e., \sCompCBt) written in Python (respectively, C++).
\item For Tests 4--9 and 21--26, the C++ version runs averagely 6.5 times faster than the (sophisticated) Python version. When we changed the second parameter of the tests from 100 to 10000 to generate 100 times bigger fuzzy graphs, the test results (not given in Table~\ref{table: JHDHG}) showed that the C++ version runs for them averagely 8.9 times faster than the (sophisticated) Python version. This is partially affected by the time needed for reading the inputs from the hard disk. 
\end{itemize}  

The aim of Tests 10--17 and 27--34 is to check the performance of the algorithms and their implementations in C++. Some notable observations are as follows:
\begin{itemize}
\item Tests 10--15 and 27--32 differ from Tests 4--9 and 21-26, respectively, only in that the second parameter for \genTest is 1000 times bigger (in order to generate 1000 times bigger fuzzy graphs). The time consumed for a test among the former is averagely 1041 times bigger than the time consumed for the corresponding one among the latter. This shows that the program \crispbis (respectively, \crispbisCS) is scalable and its performance is consistent with the complexity $O((m\log{l} + n)\log{n})$ (respectively, $O((m\log{m}$ + $n)\log{n})$) of the algorithm \CompCBt (respectively, \sCompCBt).

\item Test~16 (respectively, 33) differs from Test~15 (respectively,~32) only in that the sixth parameter for \genTest (i.e., the parameter~$l$ for \genTestC) is $10^6$ instead of 10 (respectively,~3). The time consumed for the former is about 1.16 (respectively, 1.29) times bigger than the time consumed for the latter. This shows that the programs \crispbis and \crispbisCS are scalable with respect to that parameter.

\item Comparing the test results for \crispbis and \crispbisCS (i.e., the upper part and the lower part of the last column of Table~\ref{table: JHDHG}), it can be seen that their amounts of time consumed for the tests are comparable to each other. This means that computing the largest crisp bisimulation of a finite fuzzy graph in the setting with counting successors does not require considerably more time than in the setting without counting successors. 

\item Both the programs \crispbis and \crispbisCS efficiently solved many random tests with $n^2 > 10^{12}$ and $m*n > 10^{12}$ (using the mentioned laptop). This implies that their complexity is much smaller than $(m+n)n$. This is consistent with the established complexity orders $O((m\log{l} + n)\log{n})$ and $O((m\log{m}$ + $n)\log{n})$ of the algorithms \CompCBt and \sCompCBt, respectively. 
\end{itemize}  
}

\section{Conclusions}
\label{sec: conc}

Surprisingly, as far as we know, before the current work there were no algorithms directly formulated for computing crisp bisimulations for fuzzy structures like FTSs or fuzzy interpretations in DLs. The algorithm given by Wu {\em et al.}~\cite{DBLP:journals/fss/WuCBD18} was designed for testing crisp bisimulation for NFTSs and can be applied to testing crisp bisimulation for FTSs, but its complexity $O(m^2n^4)$ is too high (and computing the largest bisimulation is more costly). One can try to adapt the algorithms with the complexity $O(n^3)$ given by Stanimirovi{\'c} {\em et at.} in~\cite{StanimirovicSC2019} to compute the largest crisp bisimulation of a given finite fuzzy automaton. 

The objective of this article was to develop efficient algorithms for computing bisimulations for fuzzy structures, for example, by applying the ideas of the Hopcroft algorithm~\cite{Hopcroft71} and the Paige and Tarjan algorithm~\cite{PaigeT87}. The task is not trivial, for example, the ideas were already exploited for the algorithms with the complexity $O(n^3)$ given by Stanimirovi{\'c} {\em et at.}~\cite{StanimirovicSC2019}, and one may wonder whether the complexity order \modified{$O((m+n)\log{n})$} like the one of the Paige and Tarjan algorithm~\cite{PaigeT87} can be obtained for the considered problem. 

We have managed to develop efficient algorithms with the complexity $O((m\log{l} + n)\log{n})$ (respectively, $O((m\log{m} + n)\log{n})$) for computing the partition corresponding to the largest crisp bisimulation of a given finite fuzzy labeled graph, for the setting without (respectively, with) counting successors. We chose fuzzy labeled graphs as they can represent fuzzy automata, FTSs and fuzzy interpretations in DLs. Taking $l = n^2$ for the \modified{worst} case and assuming that $m \geq n$, those complexity orders can be simplified to $O(m\log^2(n))$. 

\red{We have implemented both the algorithms \CompCBt and \sCompCBt in Python and C++, where the code in Python is very similar to the pseudocode given in this paper. To check the correctness of the mentioned algorithms and implemented programs, we have also implemented in Python naive algorithms for the considered problems. Many tested have been performed and all of them passed positively. The tests also show that the algorithms \CompCBt and \sCompCBt efficiently solved many randomly generated instances of the considered problem with $n^2 > 10^{12}$. This is consistent with the established complexity orders of the algorithms. Let us emphasize that all the source codes are available online and the reported tests can easily be repeated on a Linux system just by running an available Bash script~\cite{compCBfFS-code}.}


\bibliography{BSfDL}
\bibliographystyle{plain}


\end{document}